\documentclass[]{article}

\usepackage{amsthm}
\usepackage{amsfonts}
\usepackage{fullpage}
\usepackage{pgfplots}
\usepackage{amsmath}
\usepackage{multirow}
\usepackage{booktabs}
\usepackage{enumerate}
\usepackage{listings}
\usepackage{caption}
\usepackage{paralist}
\usepackage{tabularx}
\usepackage{url}

\pgfplotsset{compat=1.16}

\usetikzlibrary{pgfplots.statistics}
\usetikzlibrary{positioning, arrows.meta}
\usepgfplotslibrary{patchplots}
\usepgfplotslibrary{fillbetween}
\usepgfplotslibrary{groupplots}

\newcommand{\defparproblem}[4]{
  \vspace{3mm}
\noindent\fbox{
  \begin{minipage}{.95\columnwidth}
  \begin{tabularx}{0.95\columnwidth}{@{\extracolsep{\fill}}lr} \textsc{#1}\\ \end{tabularx} \\
  {\bf{Input:}} #2  \\
  {\bf{Parameter:}} #3 \\
  {\bf{Goal:}} #4
  \end{minipage}
  }
  \vspace{2mm}
}

\newcommand{\defproblem}[3]{
  \vspace{3mm}
\noindent\fbox{
  \begin{minipage}{.95\columnwidth}
  \begin{tabular*}{\columnwidth}{@{\extracolsep{\fill}}lr} #1  \\ \end{tabular*}
  {\bf{Input:}} #2  \\
  {\bf{Goal:}} #3
  \end{minipage}
  }
  \vspace{2mm}
  }

\newcommand{\GNRM}{{\sc Generalized Noise Role Mining}}

\newcommand{\GenPMatrixApprox}{{\sc Generalized {${\bf P}$}-Matrix Approximation}}

\newcommand{\gd}{{\sf sd}}
\newcommand{\hd}{d_{\rm H}}

\newcommand{\sd}{{\sf sd}}

\newcommand{\brank}{\textsf{BRank}}

\newcommand{\cO}{{\mathcal O}}

\newcommand{\bff}[1]{\mathbf{#1}}

\newcommand{\card}[1]{( ,\leq, )}

\theoremstyle{plain}
\newtheorem{theorem}{Theorem}[section]
\newtheorem{lemma}[theorem]{Lemma}

\newtheorem{proposition}[theorem]{Proposition}

\newtheorem{observation}[theorem]{Observation}
\newtheorem{reduction rule}{Reduction Rule}

\usepackage{authblk}

\title{Bi-objective Optimization in Role Mining}
\author[1]{Jason Crampton}
\author[1]{Eduard Eiben}
\author[1]{Gregory Gutin}
\author[2]{Daniel Karapetyan}
\author[3]{Diptapriyo Majumdar}
\affil[1]{Royal Holloway, University of London}
\affil[2]{University of Nottingham}
\affil[3]{Indraprastha Institute of Information Technology Delhi}

\begin{document}

\maketitle

\begin{abstract}
Role mining is a technique that is used to derive a role-based authorization policy from an existing policy. Given a set of users $U$, a set of permissions $P$ and a user-permission authorization relation $UPA \subseteq U \times P$, a role mining algorithm seeks to compute a set of roles $R$, a user-role authorization relation $\it UA \subseteq U \times R$ and a permission-role authorization relation $\it PA \subseteq R \times P$, such that the composition of $\it UA$ and $\it PA$ is close (in some appropriate sense) to $\it UPA$. Role mining is therefore a core problem in the specification of role-based authorization policies.
Role mining is known to be hard in general and exact solutions are often impossible to obtain, so there exists an extensive literature on variants of the role mining problem that seek to find approximate solutions and algorithms that use heuristics to find reasonable solutions efficiently.

In this paper, we first introduce the {\em Generalized Noise Role Mining} problem (GNRM) -- a generalization of the {\sc MinNoise Role Mining} problem -- which we believe has considerable practical relevance.
In particular, GNRM can produce ``security-aware'' or ``availability-aware'' solutions.
Extending work of Fomin {et al.}, we show that GNRM is fixed parameter tractable, with parameter $r + k$, where $r$ is the number of roles in the solution and $k$ is the number of discrepancies between $\it UPA$ and the relation defined by the composition of $\it UA$ and $\it PA$.
We further introduce a bi-objective optimization variant of 
GNRM, where we wish to minimize both $r$ and $k$ subject to upper bounds $r\le \bar{r}$ and $k\le \bar{k}$, where $\bar{r}$ and $\bar{k}$ are constants. We show that the Pareto front of this bi-objective optimization problem (BO-GNRM) can be computed in fixed-parameter tractable time with parameter $\bar{r} +\bar{k}$.
From a practical perspective, a solution to BO-GNRM gives security managers the opportunity to identify a mined policy offering the best trade-off between the number of policy discrepancies and the number of roles.

We then report the results of our experimental work using the integer programming solver Gurobi to solve instances of BO-GNRM\@. 
Our key findings are that (a)~we obtained strong support that Gurobi's performance is fixed-parameter tractable, 
(b)~our results suggest that our techniques may be useful for role mining in practice, based on our experiments in the context of three well-known real-world authorization policies. 
We observed that, in many cases, our solver is capable of obtaining optimal solutions when the values of either $k$ or $r$ are small.



\bigskip

\end{abstract}

\section{Introduction}
\label{sec:intro}

Role-based access control (RBAC)~\cite{SandhuCFY96} is a mature, standardized~\cite{FerraioloSGKC01} and widely deployed means of enforcing authorization requirements in a multi-user computer system.

Authorization policies, in effect, specify which interactions are authorized between users of and resources provided by a computer system.
Such policies are used by the system's access control mechanism to control interactions between users and resources.
RBAC policies authorize users for roles and roles for resources (usually referred to as ``permissions'' in the context of RBAC).
RBAC can significantly reduce the administrative burden of specifying and maintaining authorization policies, provided the set of roles is small compared to the number of users and permissions.

The problem of identifying a suitable set of roles for an RBAC system has been studied extensively over the last 25 years.
{\em Role engineering} is a top-down approach that seeks to identify roles by decomposing and analyzing business processes~\cite{NeumannS02}.
This approach does not generally scale well and requires substantial human effort~\cite{MitraSVA16}.
{\em Role mining}, the bottom-up approach, attempts to discover a set of roles from a given authorization policy that associates users directly with permissions.
More formally, the {\sc Role Mining Problem} is defined as follows:

\defproblem{\sc Role Mining Problem (RMP)}{A set of users $U$, a set of permissions $P$, a user-permission assignment relation ${\it UPA} \subseteq U \times P$, and a natural number $r$.}{Find a set $R$ of at most $r$ roles, a user-role relation ${\it UA} \subseteq U \times R$, a role-permission assignment relation ${\it PA} \subseteq R \times P$ such that $(u, p) \in {\it UPA}$ if and only if there is $\rho \in R$ such that $(u, \rho) \in {\it UA}$ and $(\rho, p) \in {\it PA}$.}

The value of $r$, for solutions that are of practical use, will be small compared to the sizes of $U$ and $P$.
However, it may be impossible to find a solution to RMP in which $r$ is sufficiently small.
Hence, approximate solutions are often sought, in which $r$ is small and the composition of $\it UA$ and $\it PA$ is close, in some suitable sense, to ${\it UPA}$.

A substantial literature now exists on role mining.
The problem is known to be hard in general and usually impossible to solve exactly (assuming the number of roles must be small relative to the number of users), so many approximate and heuristic techniques have been developed (see the survey paper of Mitra {et al.}~\cite{MitraSVA16}).

Recent work by Fomin {et al.}~\cite{FominGP20} has shown that a particlar, well-known variant of the role mining problem is {\em fixed-parameter tractable} (FPT). 
Informally, this variant is NP-hard, like many role mining problems, {so any exact algorithm to solve the problem is unlikely to be polynomial in the size of the problem's input, unless $\textsf{P}=\textsf{NP}$.}
However, there exists an algorithm (an {\em FPT algorithm}) whose running time is exponential in some of the input parameters, but polynomial in the others.
Thus, this algorithm may well be effective if the relevant parameters are small in instances of the problem that arise in practice.

Informally, the problem considered by Fomin {et al.} takes a relation ${\it UPA} \subseteq U \times P$ and natural numbers $r$ and $k$ as input. 
The goal is to find a set of roles $R$ of cardinality less than or equal to $r$, and relations ${\it UA} \subseteq U \times R$ and ${\it PA} \subseteq R \times P$ such that $|{\it UPA} \mathbin{\Delta} ({\it UA} \circ {\it PA})| \leq k$ (where ${\it UA} \circ {\it PA}$ denotes the composition of relations $\it UA$ and $\it PA$ and $\Delta$ denotes symmetric set difference).
In other words, the composition of $\it UA$ and $\it PA$ has to be similar (as defined by $k$) to $\it UPA$.
The assumption is that $k$ and $r$ will be small parameters.
This problem has been studied by the RBAC community and is usually known as the {\sc MinNoise Role Mining Problem} (MNRP)~\cite{MitraSVA16}.

One potential problem with MNRP is that it doesn't distinguish between 
\begin{inparaenum} [(a)]\item an element that is in $\it UPA$ but not in ${\it UA} \circ {\it PA}$ (which means some user is no longer authorized for some permission), and \item an element that is in ${\it UA} \circ {\it PA}$ and not in $\it UPA$ (which means that some user is now incorrectly authorized for some permission). \end{inparaenum}
We believe that in certain situations it will be important to insist that no additional authorizations are introduced by role mining (what we will refer to as \emph{security-aware} role mining), while in other situations we may require that no authorizations are lost by role mining (\emph{availability-aware} role mining).

In this paper, we introduce the {\sc GenNoise Role Mining} (GNRM), of which MNRP is a special case.
Moreover, {\sc Security-aware Role Mining} and {\sc Availability-aware Role Mining} are also special cases. 
We extend the results of Fomin {et al.} by proving that GNRM is also FPT with parameter $k+r$.



Our other theoretical contribution is to introduce a bi-objective optimization version of GNRM, called BO-GNRM, where we wish to minimize both $r$ and $k$ subject to upper bounds $\bar{r}$ and $\bar{k}$, respectively. We show that the BO-GNRM is FPT 
with parameter $\bar{r}+\bar{k}$.
Solving BO-GNRM would allow an organization to select an appropriate solution, based on the needs to balance the number of roles against the number of deviations from the original authorization matrix. Note that in order to solve BO-GNRM we use a one-objective optimization version of GNRM called OO-GNRM.

We designed a mixed-integer formulation of OO-GNRM and built a solution method based on the Gurobi solver.
We then showed that the performance of our solution method is well-aligned with the expectations for an FPT algorithm.
Furthermore, we tested our method on real-world instances; in many cases, it proved the optimality of solutions for instances with small $k$ and/or $r$.


The remainder of this section contains essential background material and defines GNRM and its bi-objective optimization version.
%

%

\subsection{Parameterized complexity}\label{sec:parameterized-complexity}
An instance of a parameterized problem $\Pi$ is a pair $(I,\kappa)$ where $I$ is the {\em main part} and $\kappa$ is the {\em parameter}; the latter is usually a non-negative integer.  
A parameterized problem is {\em fixed-parameter tractable} (FPT) if there exists a computable function $f$ such that any instance $(I,\kappa)$ can be solved in time $\cO(f(\kappa)|{I}|^c)$, where $|I|$ denotes the size of $I$ and $c$ is an absolute constant.
An algorithm to solve the problem with this running time is called an FPT algorithm.
The class of all fixed-parameter tractable decision problems is called {{\sf FPT}}. 
The function $f(x)$ may grow exponentially as $x$ increases, but the  running time may be acceptable if $\kappa$ is small for problem instances that are of practical interest.
We adopt the usual convention of omitting the polynomial factor in $\cO(f(\kappa)|{I}|^c)$ and write $\cO^*(f(\kappa))$ instead.

\subsection{Matrix decomposition and role mining}\label{sec:Br}

A \emph{Boolean matrix} is a matrix in which all entries are either $0$ or $1$.
Let $\vee$ and $\wedge$ denote the usual logical operators on the set $\{0,1\}$.
We extend these operators to Boolean matrices in the natural way~\cite{Kim82}:
\begin{enumerate}
   \item  the sum $\bff{A}\vee \bff{B}$ of Boolean matrices $\bff{A}$ and $\bff{B}$ is computed as usual with addition replaced by $\vee$;
 \item the product $\bff{A}\wedge \bff{B}$ of Boolean matrices $\bff{A}$ and $\bff{B}$ is computed as usual with multiplication replaced by $\wedge$ and addition by $\vee$. 
 {Thus, if $\bff{C}=\bff{A}\wedge \bff{B}$ then $$c_{ij}=\bigvee_{p=1}^na_{ip}\wedge b_{pj},$$ where $n$ is the number of columns in $\bff{A}$ and the number of rows in $\bff{B}$.}
\end{enumerate}
Henceforth all matrices are Boolean, unless specified otherwise.

Any binary relation $X \subseteq Y \times Z$ may be represented by a matrix $\bff{X}$ with rows indexed by $Y$ and columns indexed by $Z$, where $\bff{X}_{ij} = 1$ iff $(i,j) \in X$.
Using matrices we can reformulate the {\sc Role Mining Problem} (RMP) as follows. 
Given a matrix $\bff{UPA}$ and an integer $r$, find a matrix 
$\bff{UA}$ with $r$ columns and a matrix  $\bff{PA}$ with $r$ rows such that $\bff{UPA}=\bff{UA}\wedge \bff{PA}$. 
Thus, role mining may be regarded as a matrix decomposition problem.
 
\subsection{Generalized Noise Role Mining}\label{sec:gnrm}

As previously noted, there is often no solution to RMP if $r$ is small. In such cases, it is helpful to consider an extension of RMP called {\sc Noise Role Mining}~\cite{Dan2017,FominGP20,LuVA08,GillisV15,UzunALV11}, where the input includes a natural number $k$ and our aim is to find a matrix ${\bf UA}$ with $r$ columns and a matrix  ${\bf UA}$ with $r$ rows such that ${\hd}({\bf UPA}, {{\bf UA}\wedge {\bf PA}}) \le k$, where ${\hd}({\bf UPA},{\bf UA}\wedge {\bf PA})$ is  the number of entries in which $\bff{UPA}$ and $\bff{UA} \wedge \bff{PA}$ differ (i.e., the \emph{Hamming distance} between them). 

{\sc Noise Role Mining} could be seen as a rather crude approach to the problem of decomposing ${\bf UPA}$, as it doesn't distinguish between zeroes in ${\bf UPA}$ being replaced with ones in ${\bf UA} \wedge {\bf PA}$ and ones being replaced with zeroes.
In the first case, a user is assigned to a permission that they didn't previously have -- a potential security problem.
In the second case, a user no longer has a permission that they had been assigned, meaning the user may not be able to perform some of their responsibilities -- an availability problem.

Thus, it will be appropriate in many cases to find $\bf UA$ and $\bf PA$ such that either security or availability, as specified by $\bf UPA$ is preserved.
Informally, a refinement of {\sc Noise Role Mining}, then, would be to define {\sc Availability-preserving Role Mining}, where we require ${\bf UPA} \leq {\bf UA} \wedge {\bf PA}$, in the sense that every entry in ${\bf UPA}$ is less than or equal to the corresponding entry in ${\bf UA} \wedge {\bf PA}$.
In other words, every permission authorized by $\bf UPA$ is also authorized by ${\bf UA} \wedge {\bf PA}$.
Similarly, we could define {\sc Security-preserving Role Mining}, where we require ${\bf UPA} \geq {\bf UA} \wedge {\bf PA}$.

An even more fine-grained problem~--~the topic of this paper~--~is {\sc Generalized Noise Role Mining} (GNRM), of which
{\sc Noise Role Mining}, {\sc Availability-preserving Role Mining} and {\sc Security-preserving Role Mining} are all special cases.
In GNRM we specify at the user level whether the decomposition into $\bf UA$ and $\bf PA$ is security-preserving, availability-preserving, neither, or both.

We now introduce some notation to enable us to express GNRM formally. For a positive integer $t$, let $[t]$ denote $\{1,2,\dots ,t\}$.
Let $\bff{A}$ and $\bff{B}$ be $m \times n$ Boolean  matrices and  $\bff{F}$ be an $m \times n$ {\em label matrix} with entries $\bff{f}_{ij}\in \{\top,\bot\}$. 
The matrix $\bff{F}$ is used to define a generalized distance metric between $\bff{A}$ and $\bff{B}$.
For any $(i,j)\in [m]\times [n]$ the $F$-{\em distance from} entry $\bff{a}_{ij}$ {\em to} entry $\bff{b}_{ij}$ of matrices $\bff{A}$ and $\bff{B}$ is 
\begin{align*}
       {\sd}_{\bff{F}}(\bff{a}_{ij},\bff{b}_{ij}) & =
  	\begin{cases}
	\infty	                          & \bff{f}_{ij}=\bot \mbox{ and } \bff{a}_{ij}\ne \bff{b}_{ij},\\
	|\bff{a}_{ij}-\bff{b}_{ij}| & \mbox{otherwise}
	\end{cases}	
\end{align*}
In other words, if we are {\em not} allowed to change $\bff{a}_{ij}$ in order to obtain $\bff{b}_{ij}$ (symbol $\bot$) and  $\bff{a}_{ij}\ne \bff{b}_{ij}$ then the $F$-distance from $\bff{a}_{ij}$ to $\bff{b}_{ij}$ is $\infty$. Otherwise, it is just $|\bff{a}_{ij}-\bff{b}_{ij}|.$

We define ${\gd}_{\bff{F}}(\bff{A},\bff{B}) = \sum_{i=1}^m \sum_{j=1}^n  {\sd}_{\bff{F}}(\bff{a}_{ij},\bff{b}_{ij})$. Thus, the {\em distance} from $\bff{A}$ to $\bff{B}$ is finite if and only if for every $(i,j)\in [m]\times [n]$ such that
$\bff{a}_{ij}\ne \bff{b}_{ij}$ we have $\bff{f}(a_{ij})=\top.$

We now define \GNRM{} as a parameterized problem.

\defparproblem{{\GNRM} (GNRM)}{An $m \times n$ user-permission assignment matrix ${\bf UPA}$, a label matrix $\bff{F}$, and integers $k\ge 0$ and $r\ge 1$.}{$k + r$}{Is there an $m \times r$ user-role assignment matrix ${\bf UA}$, an $r \times n$ role-permission assignment matrix ${\bf PA}$ such that ${\gd}_{\bff{F}}(\bff{UPA},\bff{UA}\wedge \bff{PA}) \leq k$? If the answer is yes, then return such matrices $\bff{UA}$ and $\bff{PA}.$}

Note that GNRM is parameterized by the sum $k + r$.
This is because {\sc Noise Role Mining} parameterized separately by either $k$ or $r$ is intractable: in particular, for $k=0$ we have {\sc Exact Role Mining} which is {\sf NP}-hard \cite{GregoryPJL91}; and for $r=1$, {\sc Noise Role Mining} is  {\sf NP}-hard \cite{Dan2017,GillisV15}. 

Note also that GNRM reduces to:
\begin{itemize}
 \item {\sc Noise Role Mining} if $\bff{f}_{ij}=\top$ for all $(i,j)\in [m]\times [n]$;
 \item {\sc Availability-preserving Role Mining} if $\bff{f}_{ij}=\top$ if and only if $\bff{UPA}_{ij}=0$; and
 \item {\sc Security-preserving Role Mining} if $\bff{f}_{ij}=\top$ if and only if $\bff{UPA}_{ij}=1$.
\end{itemize}
%

\subsection{Bi-objective GNRM}

Note that GNRM is a decision problem, but it is clear that in practice GNRM may be viewed as an optimization problem.  
In such an optimization problem it is natural to minimize two objective functions $r$ and $k$. It is also natural to impose upper bounds on both $r$ and $k$ as if at least one of them too large then the solution may well be of no practical interest. Thus, let $\bar{r}$ and $\bar{k}$ be upper bounds for $r$ and $k$, respectively. 

This leads to the following {\sc bi-objective GNRM} problem (BO-GNRM): minimize $r$ and $k$ subject to $r\le \bar{r}$ and $k\le \bar{k}$ such that yes-answer matrices $\bff{UA}$ and $\bff{PA}$ for GNRM exist. To state BO-GNRM more formally, we will use the terminology below, which is an adaptation of multi-objective optimization terminology, see e.g. \cite{Kais99}, for BO-GNRM. We call a pair $(r',k')$ of integers  a {\em feasible solution of BO-GNRM} if given $\bff{UPA}$, $\bff{F}$ and $r'(\ge 1)$ and $k'(\ge 0)$, the answer for GNRM is yes.
A feasible solution $(r',k')$ of BO-GNRM is {\em Pareto optimal} if there is no feasible solution $(r'',k'')$ such that either $r''<r'$ and $k''\le k'$ or $r''\le r'$ and $k''< k'.$
Formally, the goal of BO-GNRM is to find the {\em Pareto front}, which is the set of all Pareto optimal solutions.

Note that bi-objective optimization has already been used for other access control problems, see e.g. \cite{CramptonGKW17,CurreyMDG20}. 

\vspace{2mm}

\paragraph{Paper organization}The rest of the paper is organized in the following way.
In Section~\ref{sec:GNRMalgo} we describe our FPT algorithms for solving GNRM and BO-GNRM\@.
We describe and discuss our experimental results in Sections~\ref{sec:GNRM-solver},~\ref{sec:is-fpt-like-runtime}, and~\ref{sec:real-world}.
Finally, in Section~\ref{sec:conclusion}, we conclude the paper with a summary of our contributions and ideas for future work.

A preliminary version of this paper~\cite{CramptonEGKM22} has appeared in proceedings of SACMAT-2022. Sections~\ref{sec:opt}, \ref{sec:GNRM-solver}, \ref{sec:is-fpt-like-runtime} and~\ref{sec:real-world} consist of new material, not published in~\cite{CramptonEGKM22}.

\section{FPT algorithms for GNRM and BO-GNRM}
\label{sec:GNRMalgo}

Fomin et al.~\cite{FominGP20} proved that {\sc Noise Role Mining} parameterized by $k + r$ is {\sf FPT}.
We will extend this result to GNRM by reducing it to the {\sc Generalized $\bf P$-matrix Approximation} problem. 
The reduction is similar to the one used in \cite{FominGP20}. 
However, Fomin et al.~\cite{FominGP20} solved {\sc Noise Role Mining} as a decision problem where the aim is only to decide whether the given instance is a yes- or no-instance. In contrast, we solve {\GNRM} as a  problem where if the given instance is a yes-instance then a solution is also returned.

We first define {\sc Generalized $\bf P$-matrix Approximation} and prove that it is FPT.
We then explain how this problem is used to establish that GNRM is FPT.

\subsection{Generalized $\bf P$-matrix approximation}

Let $\bff{P}$ be a $p \times q$ matrix (sometimes called a {\em pattern} matrix). We say that an $m \times n$ matrix $\bff{B}$ is a {\em $\bff{P}$-matrix} if there is a partition $\{I_1,\ldots,I_p\}$ of $[m]$ and a partition $\{J_1,\ldots,J_q\}$ of $[n]$ such that for every $i \in [p]$, $j \in [q]$, $s \in I_i, t \in J_j$, we have $b_{st} = p_{ij}$. 
Note that, by definition, every set in the partitions of $[m]$ and $[n]$ is non-empty. 
(Thus, $p\le m$ and $q\le n$.) In other words, $\bff{B}$ is a $\bff{P}$-matrix if $\bff{P}$ can be obtained from $\bff{B}$ by first permuting rows and columns, then partitioning the resulting matrix into blocks such that in each block $\bff{L}$ all entries are of the same value $v(\bff{L})$ and finally replacing every block $\bff{L}$ by one entry of value $v(\bff{L})$.
%

For a example, let $\bff{P} = \begin{bmatrix} 1 & 0 \\ 1 & 1 \end{bmatrix}$. 
Then $\bff{Q}_1$ and $\bff{Q}_2$ below are both $\bff{P}$-matrices: permuting columns 2 and 3 in each matrix, then partitioning (into blocks of equal size for $\bff{Q}_1$, and between rows 1 and 2 and columns 2 and 3 for $\bff{Q}_2$) and ``contracting'' gives us $\bff{P}$.
\[ 
 \bff{Q}_1 = 
 \begin{bmatrix}
  1 & 0 & 1 & 0 \\
  1 & 0 & 1 & 0 \\
  1 & 1 & 1 & 1 \\
  1 & 1 & 1 & 1 \\
 \end{bmatrix}
\qquad
 \bff{Q}_2 = 
 \begin{bmatrix}
  1 & 0 & 1 & 0 & 0 & 0 \\
  1 & 1 & 1 & 1 & 1 & 1 \\
  1 & 1 & 1 & 1 & 1 & 1 \\     
 \end{bmatrix}
\]

\defproblem{\GenPMatrixApprox}{An $m \times n$ matrix $\bff{A}$, a label matrix $\bff{F}$, a $p \times q$  matrix $\bff{P}$, and a nonnegative integer $k$.}{Is there an $m \times n$ $\bff{P}$-matrix $\bff{B}$ such that ${\gd}_F(\bff{A}, \bff{B}) \leq k$? If the answer is yes, then return such a matrix $\bff{B}.$ }

Very informally, this problem asks whether there exists a matrix $\bff{B}$ that is  (almost) the same as $\bff{A}$ and contains the rows and columns of $\bff{P}$, and, if so, returns $\bff{B}$. 
Fomin et al.~\cite{FominGP20} used the special case of {\GenPMatrixApprox}, where $\bff{f}_{ij}=\top$ for every $(i,j)\in [m]\times [n].$
It is called the {\sc $\bf P$-Matrix Approximation} problem.  We will use the following two results by Fomin et al. \cite{FominGP20}.

\begin{observation}
\label{obs:distinct-rows-columns}
Let $\bff{P}$ be a $p \times q$ matrix. Then, every $\bff{P}$-matrix $\bff{B}$ has at most $p$ pairwise distinct rows and at most $q$ pairwise distinct columns.
\end{observation}

\begin{proposition}
\label{prop:algo-testing-P-matrix}
\sloppy
Given an $m \times n$ matrix $\bff{A}$ and a $p \times q$ matrix $\bff{P}$, there is an algorithm that runs in time $2^{p \log p + q \log q} (nm)^{\cO(1)}$ and correctly outputs whether $\bff{A}$ is a $\bff{P}$-matrix. 
\end{proposition}

If $\bff{A}$ has at most $p-1$ rows or at most $q-1$ columns, then there is no $m\times n$ matrix $\bff{B}$ that is a $\bff{P}$-matrix and ${\gd}_{\bff{F}}(\bff{A}, \bff{B}) \leq k$.
In that case, the instance is a no-instance.
Let us now assume that $\bff{A}$ has at least $p$ rows and at least $q$ columns.

The next lemma was proved in \cite{FominGP20} for  {\sc $\bf P$-Matrix Approximation}. Note that replacing $\top$ in $\bff{ f}_{ij}=\top$ by 
$\bot$ for some entries $\bff{ f}_{ij}$ will only reduce the set of yes-instances of {\GenPMatrixApprox}. Thus, the next lemma follows from its special case in \cite{FominGP20}.

\begin{lemma}
\label{lemma:safeness-distinct-row-column-rule}
If $\bff{A}$ has at least $p + k + 1$ pairwise distinct rows or at least $q + k + 1$ pairwise distinct columns, then output that $(\bff{A},\bff{F}, \bff{P},k)$ is a no-instance of {\GenPMatrixApprox}.
\end{lemma}

This lemma implies the following reduction/preprocessing rule.

\begin{reduction rule}
\label{rule:distinct-row-column-rule}
Let $\bff{A}$ be a matrix. 
If $\bff{A}$ has at least $p + k + 1$ pairwise distinct rows or at least $q + k + 1$ pairwise distinct columns, then output that $(\bff{A},\bff{F},\bff{P},k)$ is a no-instance of {\GenPMatrixApprox}.
\end{reduction rule}


To simplify an instance $(\bff{A},\bff{F}, \bff{P},k)$ of {\GenPMatrixApprox},  we can apply the following reduction rule exhaustively. 
If a row $\bff{a}_i$ is deleted by the reduction rule, the label matrix $\bff{F}$ does not change for the other rows.
This means that for the reduced instance with matrix  $\bff{A}'$, the label matrix is $\bff{F}$ restricted to the rows of  $\bff{A}'$ i.e. $\{\bff{a}_1,\ldots,\bff{a}_m\} \setminus \{\bff{a}_i \}$.
For simplicity of presentation, the label matrix for $\bff{A}'$ will still be denoted by $\bff{F}$. 

We define a second reduction rule that is used to delete superfluous identical rows and columns.

\begin{reduction rule}
\label{rule:identical-row-column-rule}
If $\bff{A}$ has at least $\max \{p, k\} + 2$ identical rows, then delete one of these identical rows.
Similarly, if $\bff{A}$ has at least $\max \{q, k\} + 2$ identical columns, then delete one of these identical columns.
\end{reduction rule}

We say two instances of {\GenPMatrixApprox} are {\em equivalent} if they are both either yes-instances or no-instances.  
Fomin {et al.} proved that any application of Reduction Rule \ref{rule:identical-row-column-rule} to an instance of {\sc $\bf P$-Matrix Approximation} returns an equivalent instance of the problem~\cite[Claim 7]{FominGP20}. 
It is easy to verify that the arguments in their proof of Claim 7 also apply to {\GenPMatrixApprox}.

Applications of the two reductions rules described above either determine that the input instance is a no-instance or produce an equivalent instance with the following properties.

\begin{lemma}
\label{lemma:preprocessing-part}
Let $(\bff{A},\bff{F}, \bff{P},k)$ be an instance of {\GenPMatrixApprox}.
Then, there exists a polynomial-time algorithm that either returns ``no-instance'' or
transforms $(\bff{A},\bff{F}, \bff{P},k)$ into an equivalent instance $(\bff{A}', \bff{F}, \bff{P},k)$ of {\GenPMatrixApprox}. Moreover the following properties are satisfied.
\begin{enumerate}
	\item\label{preprocessing-property:number-rows-and-columns} The matrix $\bff{A}'$ has at least $p$ rows, at least $q$ columns, at most $(\max\{p,k\}+1)(p+k)$ rows and at most $(\max\{p,k\}+1)(p+k)$ columns.
	\item\label{preprocessing-property:poly-time} Given a $\bff{P}$-matrix $\bff{B}'$ such that ${\gd}_{\bff{F}}(\bff{A}', \bff{B}') \leq k$, in polynomial time we can compute a $\bff{P}$-matrix $\bff{B}$ such that ${\gd}_{\bff{F}}(\bff{A}, \bff{B}) \leq k$.
\end{enumerate}
\end{lemma}
\begin{proof}
Let $(\bff{A}, \bff{P}, k)$ be an input instance of {\GenPMatrixApprox}.
As $m \geq p$ and $n \geq q$, if $\bff{A}$ has at most $p-1$ rows or has at most $q-1$ columns, then there is no $m \times n$ $\bff{P}$-matrix $\bff{B}$ such that ${\gd}_{\bff{F}}(\bff{A}, \bff{B}) \leq k$.
In such a case, we return ``no-instance.''
Next, we apply Reduction Rule~\ref{rule:distinct-row-column-rule} to check the number of pairwise distinct rows as well as the number of pairwise distinct columns in $\bff{A}$.
If $\bff{A}$ has $p+k+1$ pairwise distinct rows or $q + k + 1$ pairwise distinct columns, then we return ``no-instance''.
After that, we apply Reduction Rule~\ref{rule:identical-row-column-rule} exhaustively and let $\bff{A}'$ be the obtained matrix. We also obtain a stack $S$ which contains all deleted rows and columns. 

We return $(\bff{A}',\bff{F},\bff{P},k)$ as the output instance.
Clearly, $\bff{A}'$ has at most $p + k$ pairwise distinct rows and at most $q + k$ pairwise distinct columns. Moreover, $\bff{A}'$ has at least $p$ rows and at least $q$ columns. Also, $\bff{A}'$ can have at most $\max\{p,k\} + 1$ pairwise identical rows and at most $\max\{q, k\} + 1$ pairwise identical columns.
This means that $\bff{A}'$ has at most $(\max\{p,k\} + 1) (p + k)$ rows and at most $(\max\{q, k\} + 1)(q+k)$ columns. This completes the proof that property~(\ref{preprocessing-property:number-rows-and-columns}) holds.

Suppose that $\bff{B}'$ is a $\bff{P}$-matrix such that ${\gd}_{\bff{F}}(\bff{A}',\bff{B}') \leq k$.
Note that at any intermediate stage, when a row $\bff{r}$ (a column $\bff{c}$, respectively) was deleted from $\bff{A}$, there were at least $(\max\{p,k\} + 1)$ additional rows identical to $\bff{r}$ (at least $(\max\{q,k\} + 1)$ additional columns identical to $\bff{c},$ respectively).
Hence, in $\bff{A}'$, if a row $\bff{r}$ or a column $\bff{c}$ was deleted by Reduction Rule~\ref{rule:identical-row-column-rule}, then there are exactly $\max\{p,k\} + 1$ rows identical to $\bff{r}$  (exactly $\max\{q,k\} + 1$ columns identical to $\bff{c}$, respectively). Since ${\gd}_{\bff{F}}(\bff{A}', \bff{B}') \leq k$, at most $k$ entries with label $\top$ identical to $\bff{r}$ ($\bff{c}$, respectively) were modified in $\bff{B}'$.
Thus, $\bff{B}'$ must have at least one row identical to $\bff{r}$ (at least one column identical to $\bff{c}$, respectively) if $\bff{r}$ ($\bff{c}$, respectively) was deleted by Reduction Rule~\ref{rule:identical-row-column-rule}.
Therefore, the deleted rows (columns, respectively) are identical to some rows (columns, respectively) which are the same in  in $\bff{A}'$ and $\bff{B}'$. Thus, reinstating the deleted rows and columns using stack $S$, we obtain matrices $\bff{A}$
and $\bff{B}$ such that $\bff{B}$ is a $\bff{P}$-matrix and ${\gd}_{\bff{F}}(\bff{A}, \bff{B}) \leq k$.
\end{proof}

%

\begin{theorem}
\label{lemma:FPT-P-matrix-approx}
\sloppy
{\GenPMatrixApprox} can be solved in time $2^{p\log p + q \log q}(nm)^{\cO(1)} ((\max\{p, k \} + 1)(p + k)(\max\{q, k\} + 1)(q + k))^k$.
\end{theorem}
\begin{proof}
Let $(\bff{A}, \bff{F},\bff{P}, k)$ be an instance of {\GenPMatrixApprox}.
First, we invoke the polynomial-time algorithm of Lemma~\ref{lemma:preprocessing-part} to either determine that the input instance is a no-instance or
generate an instance $(\bff{A}',\bff{F}, \bff{P}, k)$ satisfying properties~(\ref{preprocessing-property:number-rows-and-columns}) and~(\ref{preprocessing-property:poly-time}).
Recall that the first property says that $\bff{A}'$ has  at most $(\max\{p, k \} + 1)(p + k)$ rows, and  at most $(\max\{q, k\} + 1)(q + k)$ columns.
This means that $\bff{A}'$ has at most $(\max\{p, k \} + 1)(p + k)(\max\{q, k\} + 1)(q + k)$ entries.
We then consider all possible sets of at most $k$ entries.
For every entry of such a set, if the label of an entry is $\top$, we will modify it.
This results in a modified matrix $\bff{B}'$.
We then invoke Proposition~\ref{prop:algo-testing-P-matrix} to check whether $\bff{B}'$ is a $\bff{P}$-matrix or not.
This checking takes $2^{p\log p + q \log q}(nm)^{\cO(1)}$-time.
If $\bff{B}'$ is a $\bff{P}$-matrix, it is  a solution to {\GenPMatrixApprox} for the instance $(\bff{A}', \bff{F}, \bff{B}', k)$ as ${\gd}_{\bff{F}}(\bff{A}', \bff{B}') \leq k$ (since we changed at most $k$ entries in $\bff{A}'$).
Then, we make use of property~(\ref{preprocessing-property:poly-time}) to construct $\bff{B}$ satisfying ${\gd}_{\bff{F}}(\bff{A}, \bff{B}) \leq k$ and return $\bff{B}$ as a solution of {\GenPMatrixApprox} for the instance $(\bff{A}, \bff{F}, \bff{B}, k)$.
Recall that this step takes polynomial time.
Hence, the overall algorithm takes $2^{p\log p + q \log q}(nm)^{\cO(1)} ((\max\{p, k \} + 1)(p + k)(\max\{q, k\} + 1)(q + k))^k$ time. 
\end{proof}

\subsection{GNRM is FPT}
We now explain how the algorithm for {\GenPMatrixApprox} is used to solve GNRM and thus show it is FPT.
The basic strategy is to consider all possible pairs of matrices whose product $\bff{P}$ could provide the basis for a solution to GNRM.
The number of such pairs is bounded above by a function of $r$.
For each such $\bff{P}$, we determine whether the {\GenPMatrixApprox} instance $(\bff{UPA},\bff{F}, {\bf P}, k)$ has a solution, in which case we can then compute a solution to the GNRM instance.

\begin{lemma}
\label{lemma:computing-original-decomposition}
Let $\bff{P}$ be a $p\times q$ matrix such that $\bff{P} = \bff{X}\wedge \bff{Y}$ for a $p \times r$ matrix $\bff{X}$ and an $r \times q$ matrix $\bff{Y}$.
Furthermore, consider an $m \times n$ matrix $\bff{B}$ which is a $\bff{P}$-matrix.
Then, we can in polynomial time obtain an $m \times r$ matrix $\bff{X}^*$, and $r \times n$ matrix $\bff{Y}^*$ such that $\bff{B} = \bff{X}^* \wedge \bff{Y}^*$.
\end{lemma}
\begin{proof}
As $\bff{B}$ is a $\bff{P}$-matrix, there are partitions $\{I_1,\ldots,I_p\}$ of $[m]$ and $\{J_1,\ldots,J_q\}$ of $[n]$ such that for every $i \in [p], j \in [q]$, $s \in I_i, t \in I_j$, $\bff{b}_{st} = \bff{p}_{ij}$.

We initialize $\bff{X}^* = \bff{X}$ and $\bff{Y}^* = \bff{Y}$.
Consider an entry $\bff{p}_{ij}$. Let $\bff{x}_i$ be the $i$'th row of $\bff{X}$ and $\bff{y}^j$ the $j$'th column of $\bff{Y}$; then $\bff{x}_i\wedge \bff{y}^j=\bff{p}_{ij}$.
Let $c \in [p]$ and $d \in [q]$ such that $i \in I_c$ and $j \in J_d$.
Then, for any $s \in I_c$ and $t \in J_d$, set $b_{st} = \bff{p}_{ij}$.
Then, for any $s \in I_c$ and for any $t \in I_d$, we insert $\bff{x}_i$ as the $s$'th row of $\bff{X}^*$ and $\bff{y}^j$ as the $t$'th column of $\bff{Y}^*$.
\end{proof}

The {\em Boolean rank} of a matrix $\bf A$, denoted  ${\brank}(\bff{A})$, is the minimum natural number $r$ such that ${\bf A}={\bf B}\wedge {\bf C}$, where ${\bf B}$ and ${\bf C}$ are matrices such that 
the number of columns in $\bf B$ and the number of rows in $\bf C$ is $r.$
Thus, a matrix $\bff{A}$ has Boolean rank 1 if and only if $\bff{ A}=\bff{ x} \wedge \bff{ y}^T$ for some column-vectors $\bff{ x}$ and $\bff{ y}$. In fact, ${\brank}(\bff{A})=r$ if and only if $r$ is 
the minimum natural number such that ${\bf A}={\bf X}^{(1)}\vee \dots \vee {\bf X}^{(r)},$ where matrices $\bff{ X}^{(1)}, \dots , \bff{ X}^{(r)}$ are of Boolean rank 1~\cite{Kim82}.
%
%

\begin{theorem}
\label{thm:upper-role-mining-final-result-1}
\sloppy
{\GNRM} admits an $\cO^*(2^{\cO(r2^r+rk)})$-time algorithm.
\end{theorem}

\begin{proof}
Let $\bff{B}$ be an $m \times n$-matrix and let $r$ be the Boolean rank of $\bff{B}$.
Thus, there are $r$ matrices $\bff{B}^{(1)},\dots , \bff{B}^{(r)}$, each of Boolean rank 1, such that $\bff{B}=\bff{B}^{(1)}\vee \dots \vee \bff{B}^{(r)},$ where for each $i\in [r]$, $\bff{B}^{(i)}=\bff{x}^{i}\wedge (\bff{y}^{i})^T$ for some column-vectors $\bff{x}^{i}$ and $\bff{y}^{i}$.
It can be shown by induction on $r$ that $\bff{B}$ has at most $2^r$ distinct rows and at most $2^r$ distinct columns.\footnote{For $r=1$, since $\bff{B}^{(1)}=\bff{x}^{1}\wedge (\bff{y}^{1})^T$, $\bff{B}^{(1)}$ has at most two distinct rows, $\bff{y}^{1}$ and the all-zero row, and has at most two distinct columns, $\bff{x}^{1}$ and the all-zero column. Now let $r\ge 2$. By induction hypothesis, $\bff{B}=\bff{B}^{(\le r-1)}\vee \bff{B}^{(r)},$ where $\bff{B}^{(\le r-1)}$ has at most $2^{r-1}$ rows and columns and $\bff{B}^{(r)}$ at most two rows and columns. Since every row (column, respectively) of $\bff{B}$ is the disjunction of the corresponding rows (columns, respectively) of $\bff{B}^{(\le r-1)}$ and $\bff{B}^{(r)}$, the number of distinct rows (columns, respectively) in $\bff{B}$  is at most $2^{r-1}\cdot 2=2^r.$ 
}
Therefore, $\bff{B}$ is of Boolean rank at most $r$ if and only if
there is a $p\times q$ matrix $\bff{P}$ of Boolean rank at most $r$ for  $p = \min\{2^r,m\}$ and $q = \min\{2^r,n\}$ such as $\bff{B}$ is a $\bf P$-matrix.

Moreover, an $m \times n$-matrix $\bf B$ is of rank $r$ if $r$ is the minimum natural number such that $\bff{B}=\bff{C}\wedge \bff{D},$ where $\bff{C}$ is an $m \times r$-matrix and $\bf C$ is an $r \times n$-matrix. 
Hence, {\GNRM} can be reformulated as follows: Decide whether  there is a $p\times q$-pattern matrix $\bff{P}$ of
Boolean rank $r$ and an $m \times n$ $\bff{P}$-matrix $\bff{B}$ such that ${\gd}_{\bff{F}}(\bff{UPA},  \bff{B}) \leq k$ and if $\bff{B}$ does exist then find matrices $\bff{UA}$ and $\bff{PA}$ of sizes $m\times r$ and $r\times n$, respectively, such that $\bff{B}=\bff{UA}\wedge \bff{PA}.$

Thus, to solve {\GNRM} with input $(\bff{UPA},\bff{F},k)$, we can use the following algorithm: 
\begin{description}
\item[1.] Generate all pairs $(\bff{X},\bff{Y})$ of matrices of sizes $p\times r$ and $r\times q$, respectively, and for each such pair compute $\bf P=\bff{X} \wedge \bff{Y}$;
\item[2.] For each $\bff{ P},$ solve {\GenPMatrixApprox} for the instance $(\bff{UPA},\bff{F}, \bff{P}, k).$ If $(\bff{UPA}, \bff{F}, \bff{ P}, k)$ is a yes-instance, then using the algorithm of Lemma \ref{lemma:computing-original-decomposition} return matrices  $\bff{UA}$ and $\bff{PA}$ of sizes $m\times r$ and $r\times n$
such that ${\bf B}=\bff{UA}\wedge \bff{PA},$ where $\bf B$ is the solution of the instance $(\bff{UPA},\bff{F}, \bff{ P}, k)$;
\item[3.] If all instances above are no-instances of {\GenPMatrixApprox}, return ``no-instance.''
\end{description}
 It remains to evaluate the running time of the above algorithm.  
Since $p \leq 2^r$ and $q \leq 2^r$, there are at most $2^{\cO(r2^r)}$ pairs $(\bff{X},\bff{Y}),$ and we can compute all matrices $\bff{P}$ in time $2^{\cO(r2^r)}$.
Thus, the running time of the algorithm is dominated by that of Step 2. The running time of Step 2 is upper bounded by the number of matrices $\bff{P}$ (it is equal to $2^{\cO(r2^r)}$) times the maximum running time of  solving {\GenPMatrixApprox} on an instance $(\bff{UPA},\bff{F}, \bff{P}, k)$ and computing $\bff{UA}$ and $\bff{PA}$, if the instance is a yes-instance. By Lemma \ref{lemma:computing-original-decomposition}, Theorem \ref{lemma:FPT-P-matrix-approx} and the bounds $p\le 2^r, q\le 2^r$, the maximum running time is upper bounded by $\cO^*(2^{\cO(r2^r+rk)})$. It remains to observe that $2^{\cO(r2^r)}\cdot \cO^*(2^{\cO(r2^r+rk)})=\cO^*(2^{\cO(r2^r+rk)}).$
\end{proof}

\subsection{Solving BO-GNRM}\label{sec:opt}

To design an algorithm for computing the Pareto front of BO-GNRM, let us consider a related problem, the {\sc one-objective GNRM} problem (OO-GNRM): compute the minimum value $k_{\min}(r)$ of $k$ for every $r\in [\bar{r}]$ such that $k_{\min}(r)\le \bar{k}$.
In other words, given $r$, $k_{\min}(r)$ is the smallest number of discrepancies for a solution containing $r$ roles.
OO-GNRM can be easily solved by running the $\cO^*(2^{\cO(\bar{r}2^{\bar{r}}+\bar{r}\bar{k})})$-time algorithm of Theorem~\ref{thm:upper-role-mining-final-result-1} for $k\in \{0,1,\dots ,\bar{k}\}.$\footnote{We could introduce a different one-objective optimization version of GNRM, where we minimize $r$. Our choice of minimizing $k$ is explained in Section~\ref{sec:choice-of-oo}.}
(Note that if $k_{\min}(r) > \bar{k}$ then there is no solution to BO-GNRM for that value of $r$.)
This allows us to compute the Pareto front $\hat{P}$ of BO-GNRM as follows: 
\begin{equation}
\label{eq:bo-gnrm}
\hat{P} = \{(r, k_{\min}(r)):\ r \in [\bar{r}],\ 0\le k_{\min}(r) \le \bar{k},\ k_{\min}(r) < k_{\min}(r-1) \text{ if } r \ge 2\}
\end{equation}

Figure~\ref{fig:PF} illustrates the notions introduced above for $\bar{k} = 6$ and $\bar{r}=11$.
In particular, there is no solution such that $r=1$ and $k$ is less than the maximum value allowed ($6$ in this case), and no solution for $r=2$. In contrast, we can find solutions for $r \in \{3,4,5\}$ with $k_{\min}(r) = 5$. Hence $(3,5)$ belongs to the Pareto front. Similarly $(6,3)$, $(8,2)$ and $(9,0)$ belong to the Pareto front. An organization can decide which point on the Pareto front is preferable: for example, having no discrepancies between the original policy and the mined policy using nine roles versus only six roles but accepting three discrepancies.

By the arguments above, we have the following:

\begin{theorem}
\label{thm:BO}
\sloppy
There is an $\cO^*(2^{\cO(\bar{r}2^{\bar{r}}+\bar{r}\bar{k})})$-time algorithm for constructing the Pareto front of 
{BO-GNRM}. Thus, BO-GNRM is FPT with parameter  $\bar{r}+\bar{k}.$
\end{theorem}

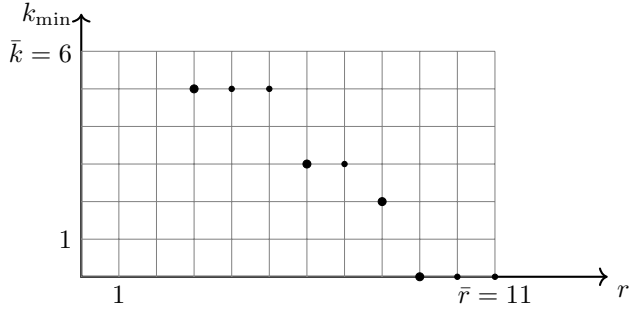
\begin{figure}
\begin{tikzpicture}[scale=0.5]
\draw [thick, <->] (0,7) -- (0,0) -- (14,0);
\node [below right] at (14,0) {$r$};
\node [below] at (11,0) {$\bar{r}=11$};
\node [below] at (1,0) {1};
\node [left] at (0,7) {$k_{\min}$};
\node [left] at (0,6) {$\bar{k}=6$};
\node [left] at (0,1) {1};
\draw [help lines] (0,0) grid (11,6);
\draw[fill] (3,5) circle [radius=3pt];
\draw[fill] (4,5) circle [radius=2pt];
\draw[fill] (5,5) circle [radius=2pt];
\draw[fill] (6,3) circle [radius=3pt];
\draw[fill] (7,3) circle [radius=2pt];
\draw[fill] (8,2) circle [radius=3pt];
\draw[fill] (9,0) circle [radius=3pt];
\draw[fill] (10,0) circle [radius=2pt];
\draw[fill] (11,0) circle [radius=2pt];
\end{tikzpicture}
\caption{Solutions of OO-GNRM and BO-GNRM,
where all the circles are solutions of OO-GNRM and all the large circles form $\hat{P}$. Note that points $(1,k_{\min}(1))$ and $(2,k_{\min}(2))$ are not depicted since $k_{\min}(2)>\bar{k}$.}
\label{fig:PF}
\end{figure}

\section{GNRM solver}
\label{sec:GNRM-solver}

To solve the BO-GNRM problem, we use a general-purpose solver.
We could have implemented a bespoke FPT algorithm to solve BO-GNRM, based on our results in the preceding sections.
However, we believe using a general-purpose solver is likely to be more useful in practice.
First, the formulation of the problem as an integer program is quite intuitive, and therefore easier to understand and maintain than a bespoke algorithm that relies on some relatively complex theory.
Second, general-purpose solvers may perform well on instances of a hard problem that is known to be FPT~\cite{Karapetyan2019}.
For example, an intelligent general-purpose solver might be able to automatically identify and apply reductions during the pre-solve process leading to an FPT-like behaviour.

In this section, we describe our approach to solving BO-GNRM, the integer programming formulation of OO-GNRM, and in Section~\ref{sec:is-fpt-like-runtime} we confirm that the empirical behaviour of this solver is consistent with that expected of an FPT algorithm.
We then apply our solver in Section~\ref{sec:real-world} to real-world instances to solve BO-GNRM.

\subsection{The choice of OO-GNRM}
\label{sec:choice-of-oo}

Our approach to solving the BO-GNRM is to decompose it into multiple OO-GNRM instances.
We have two options: (i) find $k_{\min}(r)$ for each $r$, or (ii) find $r_{\min}(k)$ for each $k$.
(Here, $r_{\min}(k)$ is the smallest value of $r$ that makes the instance of GNRM satisfiable for a given $k$.)

Note that the solution size (the number of decision variables in the solution representation) depends on $r$ but does not depend on $k$.
This makes it easier technically and computationally to fix $r$ and minimise $k$.
Also note that the value of $k_{\min}(r) \in \cO(mn)$ whereas $r_{\min}(k) \in \cO(m)$.
Indeed, we observed in our experiments that $k_{\min}(r)$ can reach large values for small $r$ whereas $r_{\min}(k)$ is relatively small even for $k = 0$.
Considering the above observations, we chose approach (i), i.e.\ our solver takes $r$ as a parameter and searches for $k_{\min}(r)$.
Then, to obtain the Pareto front for the BO-GNRM, we use formula~(\ref{eq:bo-gnrm}).

\subsection{Formulation of GNRM}

We present a mixed integer programming formulation of the OO-GNRM in order to solve the problem using a general-purpose solver.
Our formulation consists of two matrices of Boolean variables $\mathit{ua}_{i,\ell}$, $i \in [m]$, $\ell \in [r]$, and $\mathit{pa}_{\ell, j}$, $\ell \in [r]$, $j \in [n]$.
We also use a matrix of auxiliary Boolean variables $d_{i,j}$, $i \in [m]$, $j \in [n]$, representing the discrepancies between $\bff{UA} \wedge \bff{PA}$ and $\bff{UPA}$.
We use the values $0$ and $1$ to represent the values $\bot$ and $\top$, respectively, in the $\bff{F}$ matrix.
Figure~\ref{fig:formulation} describes the formulation in detail.

A na\"ive formulation of the problem would include a matrix of Boolean variables~--~to represent the product $\bff{UA} \wedge \bff{PA}$~--~and link these variables to the $\mathit{ua}_{i,\ell}$ and $\mathit{pa}_{\ell, j}$ variables.
Then it would be easy to formulate the constraints and link the decision variables to variables $d_{i,j}$.
However, knowing the values of $\textit{upa}_{i,j}$ and $f_{i,j}$ for a specific pair $(i,j)$, we can formulate the constraints more compactly.
Thus, the formulation described in Figure~\ref{fig:formulation} defines the constraints separately for each combination of values of $\textit{upa}_{i,j}$ and $f_{i,j}$.%
\footnote{The compact formulation performed significantly better than the na\"ive formulation in our experiments, so we only report results for the compact configuration.
}

\newcommand\AddLabel[0]{\refstepcounter{equation}(\theequation)}
\newcommand\commentwidth{20em}

\begin{figure*}[h]
\everymath{\displaystyle}
\renewcommand{\arraystretch}{1.3}
\frame{
\begin{tabular}{llr @{\qquad} p{\commentwidth}}
\multicolumn{2}{l}{Minimise $\sum_{i \in [m]} \sum_{j \in [n]} d_{i,j}$,}
    & \AddLabel
    & The objective is to minimise the discrepancies. \\
\multicolumn{3}{l}{For all $i \in [m]$ and $j \in [n]$ such that $f_{i,j} = 0$ and $\mathit{upa}_{i,j} = 0$:}
	& \multirow{2}{\commentwidth}{Since the discrepancy is not allowed, either $\mathit{ua}_{i,\ell}$ or $\mathit{pa}_{\ell,j}$ has to be zero for every $\ell$.} \\
	$\qquad \mathit{ua}_{i,\ell} + \mathit{pa}_{\ell,j} \le 1$
        & $\forall \ell \in [r]$, 
        & \AddLabel \\[2ex]
\multicolumn{3}{l}{For all $i \in [m]$ and $j \in [n]$ such that $f_{i,j} = 0$ and $\mathit{upa}_{i,j} = 1$:}
	& \multirow{4}{\commentwidth}{We need $\mathit{ua}_{i,\ell} = \mathit{pa}_{\ell,j} = 1$ for some $\ell$.  We enforce that at least one of $x_{i,j,1}..x_{i,j,r}$ is 1 and also if $x_{i,j,\ell} = 1$ then $\mathit{ua}_{i,\ell} = \mathit{pa}_{\ell,j} = 1$.} \\
    $\qquad x_{i,j,\ell} \le \mathit{ua}_{i,\ell}$
        & $\forall \ell \in [r]$,
        & \AddLabel \\
    $\qquad x_{i,j,\ell} \le \mathit{pa}_{\ell,j}$
        & $\forall \ell \in [r]$,
        & \AddLabel \\
    $\qquad \sum_{\ell \in [r]} x_{i,j,\ell} \ge 1$
        & $\forall \ell \in [r]$,
        & \AddLabel \\[4ex]
\multicolumn{3}{l}{For all $i \in [m]$ and $j \in [n]$ such that $f_{i,j} = 1$ and $\mathit{upa}_{i,j} = 0$:}
	& \multirow{2}{\commentwidth}{If both $\mathit{ua}_{i,\ell}$ and $\mathit{pa}_{\ell,j}$ are ones for some $\ell$ then this is a discrepancy.} \\
    $\qquad \mathit{ua}_{i,\ell} + \mathit{pa}_{\ell,j} \le 1 + d_{i,j}$
        & $\forall \ell \in [r]$,
        & \AddLabel \\[2ex]
\multicolumn{3}{l}{For all $i \in [m]$ and $j \in [n]$ such that $f_{i,j} = 1$ and $\mathit{upa}_{i,j} = 1$:}
	& \multirow{4}{\commentwidth}{If either $\mathit{ua}_{i,\ell} = 0$ or $\mathit{pa}_{\ell,j} = 0$ for every $\ell$, this is a discrepancy.  Auxiliary variable $x_{i,j,\ell}$ is forced to 0 if either $\mathit{ua}_{i,\ell} = 0$ or $\mathit{pa}_{\ell,j} = 0$.  If $x_{i,j,\ell} = 0$ for every $\ell$ then we force $d_{i,j} = 1$.} \\
	$\qquad \mathit{ua}_{i,\ell} \ge x_{i,j,\ell}$
        & $\forall \ell \in [r]$,
        & \AddLabel \\
    $\qquad \mathit{pa}_{\ell,j} \ge x_{i,j,\ell}$
        & $\forall \ell \in [r]$,
        & \AddLabel \\
    $\qquad \sum_{\ell \in [r]} x_{i,j,\ell} \ge 1 - d_{i,j}$
        & $\forall \ell \in [r]$,
        & \AddLabel \\[4ex]
%
$x_{i,j,\ell} \in \{ 0, 1 \}$
	& $\forall i \in [m],\ \forall j \in [n],\ \forall \ell \in [r]$,
    & \AddLabel 
    & Auxiliary variables, see the cases where $\mathit{upa}_{i,j} = 1$. \\
$d_{i,j} \in \{ 0, 1 \}$
	& $\forall i \in [m],\ \forall j \in [n]$,
    & \AddLabel 
    & Indicates whether there is a discrepancy in the corresponding element. \\
$\mathit{ua}_{i,j} \in \{ 0, 1 \}$
	& $\forall i \in [m],\ \forall j \in [r]$,
    & \AddLabel 
    & Defines the $\bff{UA}$ matrix. \\
$\mathit{pa}_{i,j} \in \{ 0, 1 \}$
	& $\forall i \in [r],\ \forall j \in [n]$.
    & \AddLabel
    & Defines the $\bff{PA}$ matrix.
\end{tabular}
}
\caption{CSP formulation of GNRM.}
\label{fig:formulation}
\end{figure*}

\subsection{The choice of the solver}

For our conference paper, we used the CP-SAT solver from Google OR-Tools to solve the GNRM problem (the decision version of BO-GNRM)\@.
For this paper, we considered three options:
\begin{enumerate}
    \item
    Solving GNRM using CP-SAT to identify the minimum value of $k$ that makes the instance satisfiable (\url{https://developers.google.com/optimization/cp/cp_solver});

    \item 
    Solving OO-GNRM using the linear optimisation solver from Google OR-Tools (\url{https://developers.google.com/optimization/lp}); and

    \item 
    Solving OO-GNRM using the Gurobi mixed integer programming solver (\url{https://www.gurobi.com/solutions/gurobi-optimizer}).
\end{enumerate}
Following experimentation, we concluded that Gurobi is considerably faster than the other approaches.
For example, for an instance of size $m \times n = 25 \times 25$ and $r = 5$ (with $k_\text{min}(5) = 22$), the Gurobi solver was 86 times faster than the CP-SAT-solver-based approach and 50 times faster than the linear optimisation solver from OR-Tools.
Thus, all the reported experiments in this paper are conducted with Gurobi.
We also attempted a few modifications of the formulation in Figure~\ref{fig:formulation}.
Specifically, we tried several approaches to symmetry breaking as well as adding simple custom cuts.
None of these changes improved the performance though; we assume that the internal mechanisms of Gurobi are intelligent enough to identify all the simple properties of our formulation and exploit them effectively.
We also found the default parameter values of Gurobi to be effective.
Apart from the time limit, the only Gurobi parameter that we adjusted in some experiments was \emph{MIPFocus}; we set it to 1 to intensify the search for feasible solutions when the solver was used as a heuristic.

\section{Does the Gurobi-based solver have FPT-like runtime?}
\label{sec:is-fpt-like-runtime}

\newcommand{\fixedsigmaplot}[3]{
	\addplot [#2] table [x=k, y=time, restrict expr to domain={\thisrow{sigma}}{#1:#1}, unbounded coords=discard] {3d_fit.txt};
	\addplot [#2, dotted, domain=0:30] {#3};
}

In this section, we test the hypothesis that the Gurobi-based solver is capable of exploiting the FPT structure of GNRM\@.
Since our solver addresses the OO-GNRM, we focus on testing if its running time is FPT-like with respect to $k_{\min} = k_{\min}(r)$ when $r$ is fixed.

We say that a solver has FPT-like running time if its empirical running time scales polynomially with the size of the problem instance.
As we talk about empirical running time, we focus on `typical' instances rather than the worst case.
However, the concept of a `typical' instance is vague and brings difficulties to the experimental set-up.
In the rest of this section, we introduce a new methodology to identify the scaling of the empirical running time of an algorithm and use it to show that the Gurobi-based solver has FPT-like running time.

Note that using real-world instances in such a study is generally impractical for the following reasons:
\begin{itemize}
	\item
	Such a study requires a large number of instances whereas real-world benchmark sets are usually very limited;
	
	\item
	To draw conclusions about the scaling behaviour of a solver, we need instances with a wide range of parameters which might not be present in real-world benchmark sets; and
	
	\item
	We need to ensure that the instances have consistent difficulty.
\end{itemize}

The first two points can easily be addressed by using synthetic instances produced by a pseudo-random instance generator.
The third point, however, remains a challenge even if we use a pseudo-random instance generator.
Indeed, changes in some parameters of the generator such as the instance size may affect the hardness of the instance.
In decision problems, this behaviour can often be explained by shifting between the under- and over-subscribed instance regions.
Then, to keep the hardness of the instances consistent, it is necessary to adjust some other generator parameters.
This approach was used in several other studies, e.g.~\cite{Karapetyan2022,Crampton2021,Karapetyan2019} focusing on the phase-transition region (the region between the under- and over-subscribed instances).
However, we could not use this approach in this study as phase transition is undefined for optimisation problems.

Instead, our approach in this study relies on the median running times across wide ranges of diverse instances of similar size and value of the parameter.
Thus, we measure how the median running time scales with the instance size and parameter value.

\subsection{Pseudo-random instance generator}

We adopted a pseudo-random instance generator used in earlier experimental work by Vaidya {et al.}~\cite{Vaidya2010}.
The generator takes an integer $r_0 > 0$ as a parameter and produces an instance of $\bff{UPA} \in \mathbb{R}^{m \times n}$ by creating random $\bff{UA} \in \mathbb{R}^{m \times r_0}$ and $\bff{PA} \in \mathbb{R}^{r_0 \times n}$ matrices and multiplying them together: $\bff{UPA} = \bff{UA} \wedge \bff{PA}$.
This means that we know an upper bound on the number of roles we need to mine.
We used the following settings in our generator: 
\begin{itemize}
    \item 
    the number of roles per user was randomly chosen for each user from the interval $[0, r_u]$, where $r_u$ is an instance generator parameter; and
	
    \item 
    the number of permissions per role was randomly chosen for each role from $[0, \lfloor 0.25n \rfloor]$.
	
\end{itemize}

Thus, the parameters of our generator are as follows:
\begin{itemize}
    \item
    $m$ is the number of users in the instance;

    \item 
    $n$ is the number of permissions in the instance;
    
    \item
    $r_0$ is the overall number of roles in the $\bff{UA}$ matrix used to produce the $\bff{UPA}$ matrix.

    \item
    $r_u$ is the maximum number of roles per user; $r_u \le r_0$; and

    \item 
    $r$ is the number of roles that need to be mined; this value does not affect the $\bff{UPA}$ matrix but is used by the solver.
\end{itemize}

\subsection{Data collection process}

All our computations were performed on a machine based on two Xeon~E5-2630~v2 CPUs (2.60~GHz), with 32~GB of RAM\@.
We used Gurobi~10.0.
It was restricted to one thread for the experiments in Section~\ref{sec:is-fpt-like-runtime}, with up to 12 experiments running in parallel, whereas the number of threads was unrestricted for the experiments in Section~\ref{sec:real-world} but only one experiment was conducted at a time.

In order to prove that the Gurobi-based solver has FPT-like running time, we would need to generate instances of various sizes but with fixed parameters $r$ and $k_{\min}$, and plot the running time against the instance size.
It was easy to fix the value of $r$; we set it to 5 for all the experiments in this section.
However, $k_{\min}$ is actually not a parameter of the instance generator; it is the objective value.
Thus, we developed the following methodology.

First, we generated a large set of instances for a wide range of instance generator parameters.
By solving each instance, we obtained the value of $k_{\min}$ for that instance, which then allowed us to select instances based on their values of $k_{\min}$.
For example, we could see how the running time scaled with the size of the instance for instances with $k_{\min} = 12$.

This approach, however, has a fundamental issue.
Since we used a wide range of instance parameters, it was inevitable that some of the instances were prohibitively hard.
We could put a time limit on the solver but that would skew the results; the harder instances would not be represented in our set.
Our workaround was to add a constraint to the solver restricting $k_{\min}$ to values up to 30; this made any instances with $k_{\min}$ above 30  infeasible and so we disregarded them.
Note that this did not skew the results as the set of instances for each $k_{\min} \le 30$ was unrestricted.

As we expected that the hardness of an instance mainly depends on $r$ and $k_{\min}$ but not as much on the other instance generator parameters, restricting the values of $r$ and $k_{\min}$ was sufficient to avoid overly hard instances thus making data collection feasible.

As we mentioned above, the value of $r$ was fixed to 5.
The values of the other instance generator parameters were randomly sampled from the following ranges: $n, m \in \{ 10, 11, \ldots, 70 \}$, $r_0 \in \{ 5, 6, \ldots, 15 \}$, and $r_u \in \{ 1, 2, \ldots, 5 \}$.
The selection of the parameters was based on typical values used in the literature and to ensure that the \emph{authorisation density} (the proportion of non-zero entries in the $\bff{UPA}$ matrix) was mainly in the range 5--35\%.
(All but one of the real-world datasets commonly used in role mining research have authorisation densities less than $20\%$~\cite[Table 1]{MolloyLLMWL09}.)

We set $f_{i,j} = \mathit{upa}_{i,j}$ to prioritise the security considerations; other settings of $f_{i,j}$ are studied in Section~\ref{sec:instance-types}.

At most one instance per combination of parameters was generated.
In total, we produced 110\,697 instances of which 56\,220 had $k_{\min} \le 30$ and, thus, were included in this study.
For each instance, we recorded the running time of the solver, thus our dataset included the instance parameters, the value of $k_{\min}$ and the solver running time.
In Section~\ref{sec:runtime_model}, we describe how this dataset was used to build a model of the solver running time.
This in turn was used to support our hypothesis that the Gurobi-based solver has FPT-like running time.

\subsection{Running time model}
\label{sec:runtime_model}

Our hypothesis is that the median running time of the Gurobi-based solver can be approximated as a product of two functions: a function of $k_{\min}$ only and a function of the instance size only (as we fixed $r$ in this experiment, it is not a part of the model).
If the second function is a polynomial then we can claim that the Gurobi-based solver has FPT-like running time, meaning that it is suitable for reasonably large instances as long as the value of $k_{\min}$ is small.

We came up with two candidates for the definition of the instance size:
\begin{itemize}
    \item
    The size $mn$ of the UPA matrix; and
	
    \item
    The number of non-zero elements in the UPA matrix, which we denote by $\sigma$.
\end{itemize}

There are arguments for each of these candidate definitions; indeed, the size of the MIP formulation in Figure~\ref{fig:formulation} depends, in various ways, on both $mn$ and $\sigma$.
Thus, we tested both definitions when fitting running time models.
In our experience, the second definition gives a much better fit, hence we use $\sigma$ as the definition of the instance size in the rest of the paper.

Let $f(k_{\min}, \sigma)$ be our model of the median running time of the Gurobi-based solver.
While it is possible to fit $f(k_{\min}, \sigma)$ to all the data points, the result would be affected by the imbalances in our set of instances; for example, the lower values of $k_{\min}$ are represented better in our dataset, and this would be reflected in the model.
Thus, we designed the following process to balance the dataset.
For each $k_{\min} \in \{0, 1, \ldots, 30\}$ and $j \in \{ 1, 2, \ldots \}$, we calculated the median running time $t_{k_{\min},j}$ over all the instances of size $\sigma$ such that $100 (j-1)\le \sigma < 100 j$ and with the given value of $k$.
We then fit $f(k_{\min}, \sigma)$ to $(k_{\min}, 100j - 50, t_{kj})$ for all the combinations of $k_{\min}$ and $j$ where the number of instances within the corresponding range is at least 10.

The running times across our experiments vary from milliseconds to minutes; a residual of a few seconds is acceptable for larger and harder instances but is a poor approximation for small and easy instances.
In other words, when fitting $f(k_{\min}, \sigma)$, we are concerned with the relative error, not the absolute error.
Thus, we use the logarithmic scale for the running times, i.e.\ we fit $\log f(k_{\min}, \sigma)$ to $\log t_{k_{\min}, j}$.

According to our initial assumption, the model was a product of two functions: $f(k_{\min}, \sigma) = f_k(k_{\min}) \cdot f_\sigma(\sigma)$.
By visualising the experimental data, we established that $f_k(k_{\min})$ was exponential in $k_{\min}$ as can be clearly seen in Figure~\ref{fig:sigma-slices}, however the slope of the curves was different for different values of $\sigma$.
In other words, $f_k$ turned out to be a function of both $k_{\min}$ and $\sigma$.
By estimating the slopes, we established that a good model for $f_k$ was $f_k(k_{\min}, \sigma) = c_1 \cdot 2^{k_{\min} / (c_2 \sigma + c_3)}$.

A reasonable fit for $f_\sigma(\sigma)$ could be a parabola or an exponential function with the base close to 1.
We tested both hypotheses and found that a quadratic function fits significantly better than exponential.

Thus, our final model is as follows:
\begin{equation}
\label{eq:best-fit}
f(k_{\min}, \sigma) = c_1 \cdot 2^{k_{\min} / (c_2 \sigma + c_3)} \cdot (\sigma^2 + c_4 \sigma + c_5) \,,
\end{equation}
where $c_1, c_2, \ldots, c_5$ are coefficients.
This model gives a close fit for the experimental data except for small values of $k_{\min}$.
Thus, we removed $k_{\min} < 2$ when fitting the data and obtained the following values of the coefficients: $c_1 \approx 1.11 \cdot 10^{-4}$, $c_2 \approx 3.29 \cdot 10^{-2}$, $c_3 \approx 4.14$, $c_4 \approx -77.3$ and $c_5 \approx 7.23 \cdot 10^3$.

Figure~\ref{fig:3d-fit} shows the fit of $f(k_{\min}, \sigma)$ in three dimensions, whereas Figures~\ref{fig:k-slices} and~\ref{fig:sigma-slices} show slices through the space of $k_{\min}$ and $\sigma$.
Except for small $k_{\min}$, the model accurately predicts the aggregated running times.
It even predicts that, for small $\sigma$ and large $k_{\min}$, the running time slightly decreases as $\sigma$ increases, see Figure~\ref{fig:k-slices}.

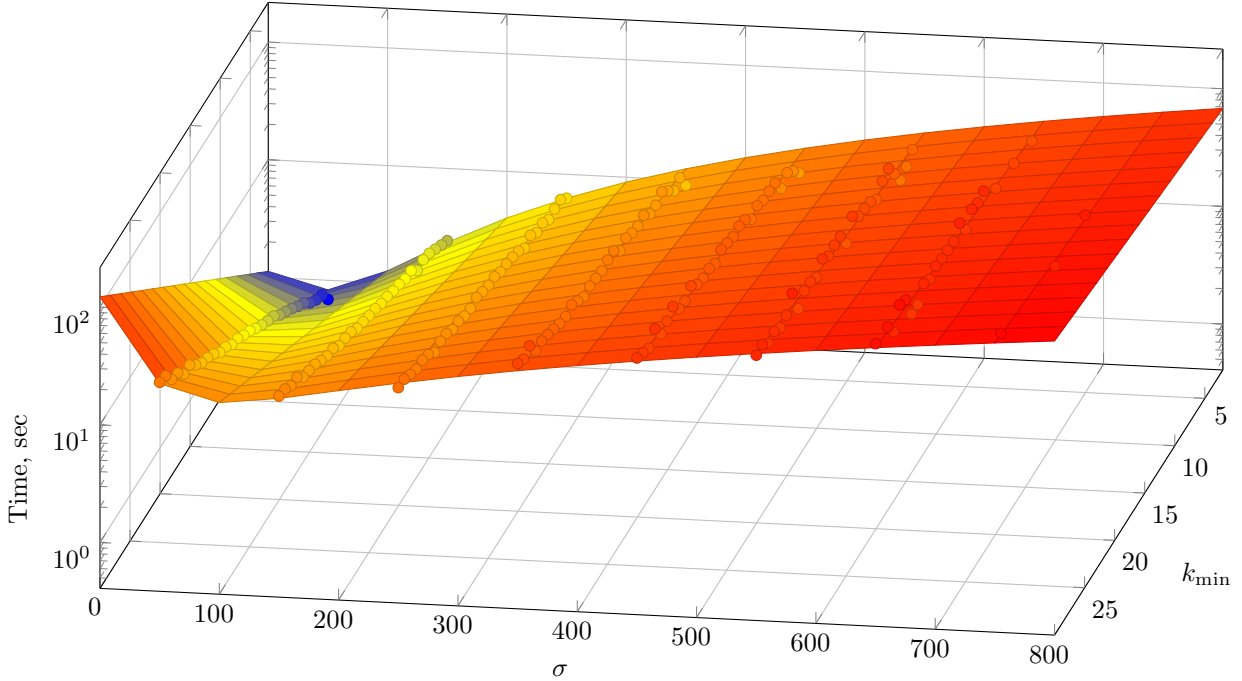
\begin{figure}[htb]
\begin{tikzpicture}
    \begin{axis}[
      xlabel={$k_{\min}$},
      ylabel={$\sigma$},
      zlabel={Time, sec},
      zmode=log,
      width=\textwidth,
      height=10cm,
      view={100}{40},
      xtick={5, 10, 15, 20, 25},
      grid=major
    ]
    \addplot3[scatter, only marks] table[x=k, y=sigma, z=time] {data/3d_fit.txt};
    \addplot3[mark=none, draw, domain=2:30, domain y=0:800, samples=17, surf, mesh/ordering=y varies, opacity=0.5] {1.11264829e-04 * (2^(x / (3.29453098e-02 * y + 4.13642674e+00))) * (y^2 + -7.73164390e+01*y + 7.22783449e+03)};
    \end{axis}
\end{tikzpicture}

\caption{This graph demonstrates how our model $f(k_{\min}, \sigma)$ (surface) fits the data aggregated into $t_{k_{\min}, j}$ (scatter plot).
The colour represents the time (the value along the vertical axis).}
\label{fig:3d-fit}
\end{figure}

\newcommand{\slicekplot}[2]{
\nextgroupplot[
    title={$k = #1$}, 
    title style={
        at={(0.5,0.0)}, 
        anchor=south, 
        inner sep=3pt, 
        fill=white, 
        draw=black, 
        rectangle
    },
    legend to name=grouplegend-#1, 
    legend columns=3, 
    legend style={column sep=1em}, 
    legend cell align=left,
    #2
]
    \addplot[only marks, mark size=1, mark options={fill=blue}] table[
        x=sigma,
        y=time,
    ] {data/slices/slice_k#1.txt};
    \addlegendentry{\hspace{-1em}Median times $t_{k,j}$}

    \addplot[red, ultra thick] table[
        x=sigma,
        y=fit,
    ] {data/slices/slice_k#1.txt};
    \addlegendentry{\hspace{-1em}$f(k, \sigma)$}
}

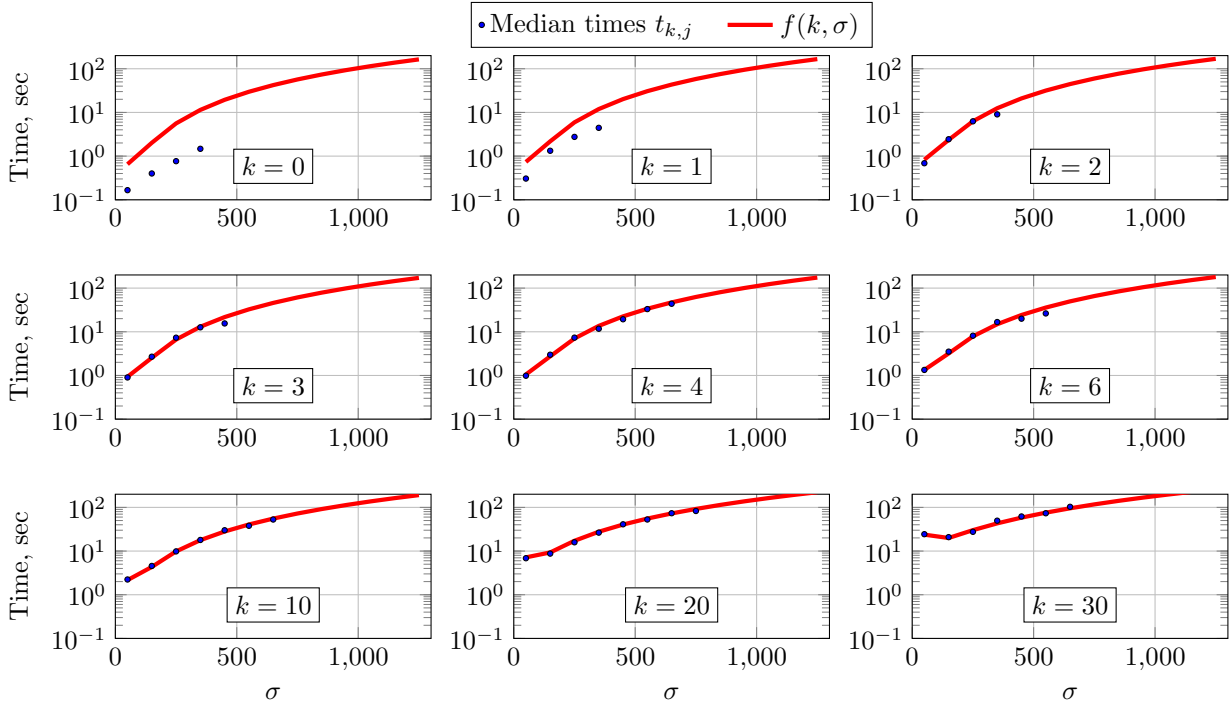
\begin{figure}[htb]
\begin{tikzpicture}
\begin{groupplot}[
    group style={
        group size=3 by 3,
        horizontal sep=1.1cm,
        vertical sep=1cm,
    },
    width=0.35\textwidth,
    height=3.5cm,
    ymode=log,
    xlabel={$\sigma$},
    ylabel={Time, sec},
    xmin=0,
    xmax=1300,
    ymin=0.1,
    ymax=200,
    grid=major
]    
    \slicekplot{0}{xlabel={}}
    \coordinate (point1) at (rel axis cs:0.5,1.05);
    \slicekplot{1}{xlabel={}, ylabel={}}
    \slicekplot{2}{xlabel={}, ylabel={}}
    \coordinate (point2) at (rel axis cs:0.5,1.05);

    \slicekplot{3}{xlabel={}}
    \slicekplot{4}{xlabel={}, ylabel={}}
    \slicekplot{6}{xlabel={}, ylabel={}}
    
    \slicekplot{10}{}
    \slicekplot{20}{ylabel={}}
    \slicekplot{30}{ylabel={}}
    
\end{groupplot}

\coordinate (midpoint) at ($(point1)!0.5!(point2)$);
\node[inner sep=0pt, anchor=south] at (midpoint) {\pgfplotslegendfromname{grouplegend-1}};
\end{tikzpicture}

\caption{The aggregated data $t_{k, j}$ and the best fit model $f(k, \sigma)$ sliced along the $k$ axis.}
\label{fig:k-slices}
\end{figure}

\newcommand{\sliceonesplot}[3]{
    \nextgroupplot[
        title={$\sigma \in [#2]$}, 
        legend to name=grouplegendones-#1, 
        legend columns=3, 
        legend style={column sep=1em}, 
        legend cell align=left,
        title style={
            at={(0.5,0)}, 
            anchor=south, 
            inner sep=3pt, 
            fill=white, 
            draw=black, 
            rectangle
        },        
        #3
    ]
    
    \addplot[only marks, mark size=1, mark options={fill=blue}] table[
        x=k,
        y=time,
    ] {data/slices/slice_ones#1.txt};
    \addlegendentry{\hspace{-1em}Median times $t_{k, j}$}

    \addplot[red, ultra thick] table[
        x=k,
        y=fit,
    ] {data/slices/slice_ones#1.txt};
    \addlegendentry{\hspace{-1em}$f(k, \sigma)$}
}

\begin{figure}[htb]
\begin{tikzpicture}
\begin{groupplot}[
    group style={
        group size=3 by 2,
        horizontal sep=1.3cm,
        vertical sep=1cm,
    },
    width=0.35\textwidth,
    height=4cm,
    ymode=log,
    xlabel={$k_{\min}$},
    ylabel={Time, sec},
    xmin=0,
    xmax=30,
    ymin=0.1,
    ymax=200,
    grid=major,
]
    
    \sliceonesplot{0}{0, 99}{xlabel={}}
    \coordinate (point1) at (rel axis cs:0.5,1.05);

    \sliceonesplot{1}{100, 199}{xlabel={}, ylabel={}}
    
    \sliceonesplot{2}{200, 299}{xlabel={}, ylabel={}}
    \coordinate (point2) at (rel axis cs:0.5,1.05);

    \sliceonesplot{3}{300, 399}{}
    
    
    \sliceonesplot{5}{500, 599}{ylabel={}}
    
    
    \sliceonesplot{7}{700, 799}{ylabel={}}
    
\end{groupplot}

\coordinate (midpoint) at ($(point1)!0.5!(point2)$);

\node[inner sep=0pt, anchor=south] at (midpoint) {\pgfplotslegendfromname{grouplegendones-1}};

\end{tikzpicture}

\caption{The aggregated data $t_{k_{\min}, j}$ and the best fit model $f(k_{\min}, \sigma)$ sliced along the $\sigma$ axis.}
\label{fig:sigma-slices}
\end{figure}
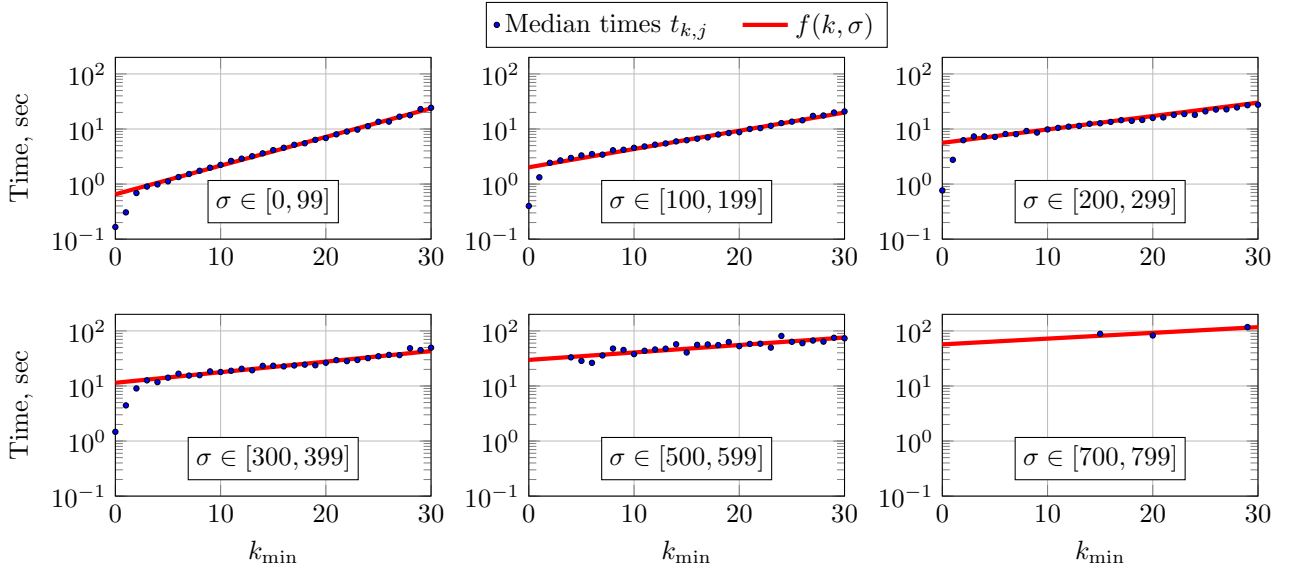

So far, we focused on the median running times.
To see how far the actual running times fall from the predicted ones, we produced one more visualisation, see Figure~\ref{fig:unified}.
It is designed to demonstrate that our $f_\sigma(\sigma)$ is a good fit for the data across all values of $k_{\min} \ge 2$.
To compensate for the differences caused by the variation of $k_{\min}$, we divided the running times by $f_k(k_{\min}, \sigma)$.
One can see that all the instances fall into a relatively narrow interval around our model $f_\sigma(\sigma)$, except for small $\sigma$.
Similarly to the case of small $k_{\min}$, we assume that small $\sigma$ may significantly change the behaviour of Gurobi.

Nevertheless, it is clear from this visualisation that a good model for $f_\sigma(\sigma)$ has to be a polynomial; an exponential dependence would be a straight line in this semi-logarithmic plot.
Our model accurately approximates the data for all $\sigma$ and $k_{\min}$ except for several smallest values.

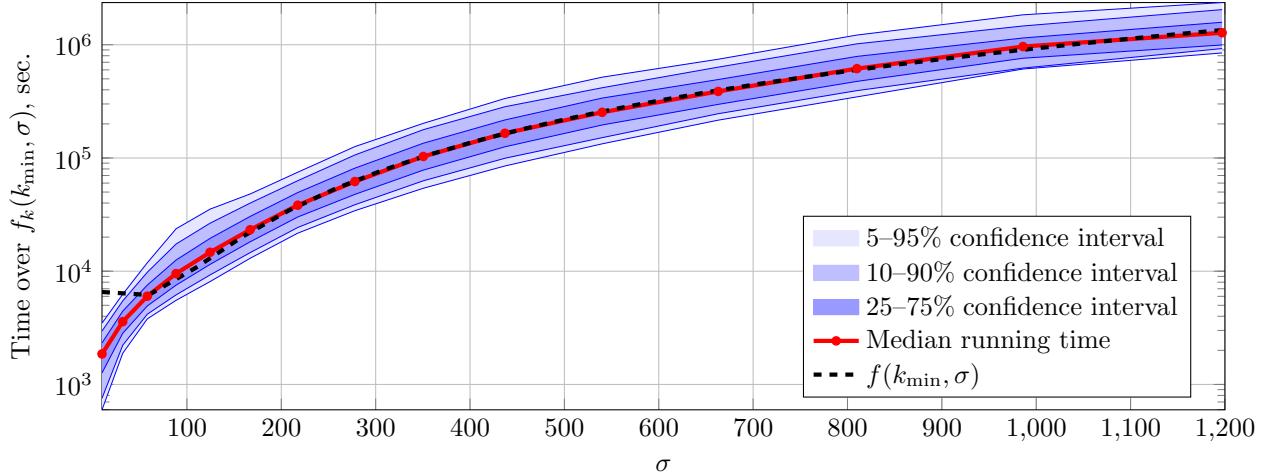
\begin{figure}[htb]
\begin{tikzpicture}
\begin{semilogyaxis}[
    xlabel={$\sigma$},
    ylabel={Time over $f_k(k_{\min}, \sigma)$, sec.},
    grid=major,
    width=\textwidth,
    height=7cm,
    every axis plot/.append style={mark=none, ultra thin},
    legend pos={south east},
    legend cell align=left,
    enlargelimits=false
]
    \addplot[blue, name path=p5, forget plot] table[x=middle, y=p5] {data/unified_percentiles.txt};
    \addplot[blue, name path=p10, forget plot] table[x=middle, y=p10] {data/unified_percentiles.txt};
    \addplot[blue, name path=p25, forget plot] table[x=middle, y=p25] {data/unified_percentiles.txt};
    \addplot[blue, name path=p75, forget plot] table[x=middle, y=p75] {data/unified_percentiles.txt};
    \addplot[blue, name path=p90, forget plot] table[x=middle, y=p90] {data/unified_percentiles.txt};
    \addplot[blue, name path=p95, forget plot] table[x=middle, y=p95] {data/unified_percentiles.txt};
    
    \addplot[blue!10] fill between[of=p5 and p10];
    \addlegendentry{5--95\% confidence interval}
    \addplot[blue!25] fill between[of=p10 and p25];
    \addlegendentry{10--90\% confidence interval}
    \addplot[blue!40] fill between[of=p25 and p75];
    \addlegendentry{25--75\% confidence interval}
    \addplot[blue!25, forget plot] fill between[of=p75 and p90];
    \addplot[blue!10, forget plot] fill between[of=p90 and p95];

    \addplot[red, ultra thick, mark=*, mark size=1pt] table[x=middle, y=p50] {data/unified_percentiles.txt};
    \addlegendentry{Median running time}
    \addplot[dashed, ultra thick, black, domain=10:1200, mark=none] {x^2 - 7.73164390e+01 * x + 7.22783449e+03};
    \addlegendentry{$f(k_{\min}, \sigma)$}
\end{semilogyaxis}
\end{tikzpicture}

\caption{This graph demonstrates the fit of $f_\sigma(\sigma)$ to the data for all $2 \le k_{\min} \le 30$.
The times are divided by $f_k(k_{\min}, \sigma)$ to compensate for the varying $k_{\min}$.
Blue areas show the distribution of the data; blue lines are percentiles of the running times.}
\label{fig:unified}
\end{figure}

\subsection{Conclusions about the FPT-like behaviour}

We conclude that~(\ref{eq:best-fit}) is a good fit to the experimental data for $k_{\min} \ge 2$.
The instances for $k_{\min} < 2$ are being solved notably faster than our model predicts.
We hypothesise that this behaviour is linked to the use of heuristics within the Gurobi solver, however further investigation is needed to confirm this.

An algorithm is formally called FPT if its worst-case time complexity is $\cO(f_k(k_{\min}) \cdot f_\sigma(\sigma))$, where $f_\sigma(\sigma)$ is a polynomial.
In our model of the \emph{median} running times, $f_\sigma(\sigma)$ is indeed polynomial but $f_k$ depends not only on $k_{\min}$ but also on $\sigma$.
However, the dependence on $\sigma$ is inverted; the larger the $\sigma$, the smaller the $f_k(k_{\min}, \sigma)$.
Thus, we can substitute $\sigma = 1$ to~(\ref{eq:best-fit}) to obtain an upper bound for the median running times of the Gurobi solver of the standard FPT form: $f(k_{\min}, \sigma) = \cO(2^{k_{\min} / (c_2 + c_3)} \cdot (\sigma^2 + c_4 \sigma + c_5))$.

An interesting research question is to understand why the exponent in $f_k(k_{\min}, \sigma)$ is inversely proportional to $\sigma$.
If we ignore $c_3$, we see that the exponent in $f_k(k_{\min}, \sigma)$ is the proportion of values in the $\bff{UPA}$ matrix that are adjusted by the role mining: $k_{\min} / \sigma$.
Our theory does not explain this phenomenon but it is an interesting observation that deserves future research; for example, there could be parameterisation of the problem based on this ratio.

Overall, these results confirm that general-purpose solvers can be effective on FPT problems. Moreover, they suggest that our solver might be appropriate for relatively large instances, provided the value of the parameter $k_{\min}$ is small enough.
Further research may build a model of the solver's running time as a function of three parameters: $\sigma$, $k_{\min}$ and $r$.

\section{Computational experiments with real-world instances}
\label{sec:real-world}

To test our solver and study the trade-off between $k_{\min}$ and $r$ in BO-GNRM, we use a set of real-world benchmark instances~\cite{Ene2008}.
This set includes nine instances of various sizes.
The largest instance \emph{americas large} has $m = 3\,485$, $n = 10\,127$, $\sigma = 185\,294$ and around 400 roles according to the best heuristical solutions in the literature.
Our approach is impractical for such large instances as the size of the formulation would be prohibitive.
However, some other instances in this benchmark set are more manageable.
The details of the instances we use in this section are given in Table~\ref{tab:instances}.

\begin{table}[htb]
\begin{tabular}{lrrrrr}
\toprule
Instance & $m$ & $n$ & $\sigma$ & density & $r_\text{min}$ \\
\midrule
Domino & 79 & 231 & 730 & 4.0\% & 20 \\
Firewall 2 & 325 & 590 & 36\,428 & 19.0\% & 10 \\
Healthcare & 46 & 46 & 1\,486 & 70.2\% & 14 \\
\bottomrule
\end{tabular}

\caption{Real-world instances used in this section.  The last column gives the minimum number of roles $r_\text{min}$~\cite{Ene2008}; this is for RMP (i.e., $k_{\min} = 0$).}
\label{tab:instances}
\end{table}

In classic role mining, the $\bff{UPA}$ matrix is sufficient to define the instance.
In GNRM, we also need to define the value of $r$ (the number of roles to be extracted) and the matrix $\bff{F}$.

In our experiments, we use three types of the $\bff{F}$ matrix:
\begin{description}
	\item[Security:] $f_{i,j} = \mathit{upa}_{i,j}$; the new roles may remove a permission from a user but can never add a new permission;
	\item[Availability:] $f_{i,j} = 1 - \mathit{upa}_{i,j}$; the new roles may add a permission to a user but can never remove an existing permission;
	\item[Noise:] $f_{i,j} = 1$ for every $i$ and $j$; the new roles may both remove and add permissions.
\end{description}

In Section~\ref{sec:performance}, we focus on the Security instances.
We compare all three types of the $\bff{F}$ matrix in Section~\ref{sec:instance-types}.

Our Gurobi-based solver can be used as an exact solver, but it can also serve as a heuristic if we specify the time budget.
In this section, we use it as a heuristic, however in some cases it proves the optimality of the solutions within the provided time budget.


\subsection{Performance of the solver}
\label{sec:performance}

Figure~\ref{fig:real-k-vs-r} shows the trade-off between $k_{\min}$ and $r$ for several real-world instances.
We experimented with three time budgets: 100 sec, 1\,000 sec and 10\,000 sec.
All the solutions that were proven optimal by Gurobi within the time budget were marked with circles.
Thus, for example, being given 10000 seconds, the solver proved optimality of solutions for $k_{\min} = 1, 2, 3, 4, 19$ for the \emph{Domino} instance.
Note that the vertical axes in these plots are logarithmic while $k_{\min}$ can take any integer value starting from $k_{\min} = 0$.
Thus, we added value $k_{\min} = 0$ to the logarithmic scale.

\pgfplotsset{
	tb100/.style={color=green},
	tb1000/.style={color=red},
	tb10000/.style={color=blue, dashed, dash pattern=on 3pt off 3pt},
	tb10000avail/.style={color=olive},
	tb10000noise/.style={color=orange, dashed, dash pattern=on 3pt off 3pt, dash phase=3pt},
	tradeoff line/.style={thick, mark=none},
	proven marks/.style={only marks, mark=o, mark options={thin, solid}},
	types line/.style={thick},
	ticks1/.style={
	    ytick=    {0.1, 1, 1e1, 1e2, 1e3, 1e4, 1e5},
    		yticklabels={0, $10^0$, $10^1$, $10^2$, $10^3$, $10^4$, $10^5$},
		minor ytick = {0,
						1e0, 2e0, 3e0, 4e0, 5e0, 6e0, 7e0, 8e0, 9e0, 
						1e1, 2e1, 3e1, 4e1, 5e1, 6e1, 7e1, 8e1, 9e1, 
						1e2, 2e2, 3e2, 4e2, 5e2, 6e2, 7e2, 8e2, 9e2, 
						1e3, 2e3, 3e3, 4e3, 5e3, 6e3, 7e3, 8e3, 9e3, 
						1e4, 2e4, 3e4, 4e4, 5e4, 6e4, 7e4, 8e4, 9e4, 
						1e5, 2e5, 3e5, 4e5, 5e5, 6e5, 7e5, 8e5, 9e5, 1e6}
	},
}

\newcommand{\kvsr}[5]{\addplot [#5] table [x=r, y expr={\thisrow{mink} == 0 ? 0.1 : \thisrow{mink}}, restrict expr to domain={\thisrow{time_budget}}{#2:#2}, restrict expr to domain={\thisrow{f_type}}{#3:#3}, restrict expr to domain={\thisrow{proven}}{#4}, unbounded coords=discard] {data/cache_gurobi_mink_#1.txt};}

\newcommand{\minkvsr}[2]{
    \nextgroupplot[
        xlabel={$r$},
        ylabel={$k_{\min}$},
        title style={
            at={(0.5,0.92)}, 
            anchor=north, 
            inner sep=3pt, 
            fill=white, 
            draw=black, 
            rectangle
        },
	legend to name=grouplegend-#1, 
	legend columns=3, 
	legend style={column sep=1em}, 
	legend cell align=left,
        #2
    ]
	\kvsr{#1}{100}{1}{0:1}{tb100, tradeoff line}
	\addlegendentry{\hspace{-1em}Time budget 100 sec}	

	\kvsr{#1}{1000}{1}{0:1}{tb1000, tradeoff line}
	\addlegendentry{\hspace{-1em}Time budget 1\,000 sec}	

	\kvsr{#1}{10000}{1}{0:1}{tb10000, tradeoff line}
	\addlegendentry{\hspace{-1em}Time budget 10\,000 sec}	

	\kvsr{#1}{100}{1}{1:1}{tb100, proven marks, mark size=3pt}
	\addlegendentry{\hspace{-1em}Proven (100 sec)}	

	\kvsr{#1}{1000}{1}{1:1}{tb1000, proven marks, mark size=2pt}
	\addlegendentry{\hspace{-1em}Proven (1\,000 sec)}	

	\kvsr{#1}{10000}{1}{1:1}{tb10000, proven marks, mark size=1pt}
	\addlegendentry{\hspace{-1em}Proven (10\,000 sec)}	
}

\begin{figure}[htb]
\begin{tikzpicture}
\begin{groupplot}[
    group style={
    group size=3 by 1,
	horizontal sep=1cm,
	vertical sep=2cm,		
    },
    width=0.37\textwidth,
    height=6cm,
    grid=major,
    ymode=log,
    ymin=0.1,
    xmin=0,
    xtick distance=5,
    minor tick num=4,
]
    \minkvsr{domino}{ticks1, title=Domino, xmax=25.5}
    \coordinate (point1) at (rel axis cs:0.5,1.05);
    
    \minkvsr{fire2}{ticks1, title=Firewall 2, ylabel={}}
    
    \minkvsr{hc}{ticks1, title=Healthcare, ylabel={}, xmax=17.5}
    \coordinate (point2) at (rel axis cs:0.5,1.05);
\end{groupplot}

\coordinate (midpoint) at ($(point1)!0.5!(point2)$);

\node[inner sep=0pt, anchor=south] at (midpoint) {\pgfplotslegendfromname{grouplegend-domino}};
\end{tikzpicture}

\caption{The trade-off between $k_{\min}$ and $r$ for several real-world instances (Security type).
Each line corresponds to the values obtained by our Gurobi-based solver within the given time budget (100, 1000 or 10000 seconds).
When the time budget was sufficient to prove optimality of the solution, a corresponding mark was added.
Note that we added value $k = 0$ to the vertical axis.
}
\label{fig:real-k-vs-r}
\end{figure}
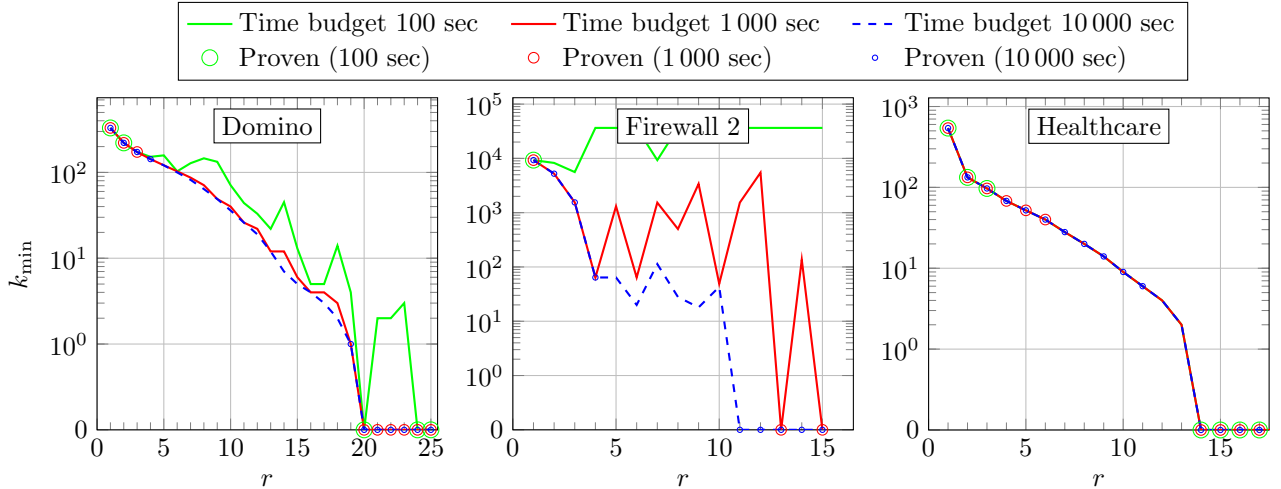

Generally, increasing the value of $r$ decreases the value of $k$; indeed, the \emph{optimal} value of $k_\text{min}(r)$ is a monotonically decreasing function of $r$.
However, there are exceptions to this rule in Figure~\ref{fig:real-k-vs-r}, particularly for smaller time budgets.
In many cases, increasing the value of $r$ makes the problem harder as the formulation size increases, and the solver might not be able to find a near-optimal solution within the time budget.
In extreme cases, the solver only finds the trivial solution $\mathit{ua}_{i, \ell} = \mathit{ua}_{\ell, j} = 0$ for every $i$, $j$ and $\ell$.

Note that, formally speaking, Figure~\ref{fig:real-k-vs-r} does not show solutions to BO-GNRM\@.
In practice, one would want to remove the dominated solutions to obtain the Pareto front and set the upper bounds $\bar{k}$ and $\bar{r}$.
These steps, however, would hide some of the aspects of the solver's behaviour, hence we decided to include all the solutions here and below.

Naturally, lower time budgets produce worse solutions in general.
However, the difference between the 1\,000~sec time budget and the 10\,000~sec time budget is often insignificant.
This hints at the proximity of the obtained solutions to the optimal solutions.

\subsection{Comparison of instance types}
\label{sec:instance-types}

So far, we experimented with the Security instances only.
In this section, we will study how changing the instance type (Security, Availability, and Noise) affects the properties of the instances.
We will use the same real-world $\bff{UPA}$ matrices but we will solve instances with various $\bff{F}$ matrices.

Figure~\ref{fig:real-instance-types} shows the trade-off between $r$ and $k$ for each instance type, for several real-world instances.
As in Section~\ref{sec:performance}, the vertical axes are logarithmic with added $k = 0$.
We use the 10\,000 seconds time budget for this experiment to be as close to the optimal solutions as we can.

\newcommand{\comparef}[2]{
    \nextgroupplot[
        title style={
            at={(0.5,0.92)}, 
            anchor=north, 
            inner sep=3pt, 
            fill=white, 
            draw=black, 
            rectangle
        },
	legend to name=grouplegendtypes-#1, 
	legend columns=3, 
	legend style={column sep=1em}, 
	legend cell align=left,
	#2
    ]
	
	\kvsr{#1}{10000}{1}{0:1}{tb10000, types line}
	\addlegendentry{\hspace{-1em}Security}

	\kvsr{#1}{10000}{2}{0:1}{tb10000avail, types line}
	\addlegendentry{\hspace{-1em}Availability}	

	\kvsr{#1}{10000}{3}{0:1}{tb10000noise, types line}
	\addlegendentry{\hspace{-1em}Noise}	

	\kvsr{#1}{10000}{1}{1:1}{tb10000, proven marks, mark size=1pt}
	\addlegendentry{\hspace{-1em}Proven (Security)}	

	\kvsr{#1}{10000}{2}{1:1}{tb10000avail, proven marks, mark size=2pt}
	\addlegendentry{\hspace{-1em}Proven (Availability)}	

	\kvsr{#1}{10000}{3}{1:1}{tb10000noise, proven marks, mark size=3pt}
	\addlegendentry{\hspace{-1em}Proven (Noise)}	
}

\begin{figure}[htb]
\begin{tikzpicture}
\begin{groupplot}[
    group style={
        group size=3 by 1,
	horizontal sep=1cm,
	vertical sep=2cm,
    },
    xlabel={$r$},
    ylabel={$k_{\min}$},
    width=0.37\textwidth,
    height=6cm,
    grid=major,
    ymode=log,
    ymin=0.1,
    xmin=0,
    xtick distance=5,
    minor tick num=4,
]

    \comparef{domino}{ticks1, title=Domino, xmax=25.5}
    \coordinate (point1) at (rel axis cs:0.5,1.05);
    
    \comparef{fire2}{ticks1, title=Firewall 2, ylabel={}}

    \comparef{hc}{ticks1, title=Healthcare, ylabel={}, xmax=17.5}
    \coordinate (point2) at (rel axis cs:0.5,1.05);
\end{groupplot}

\coordinate (midpoint) at ($(point1)!0.5!(point2)$);

\node[inner sep=0pt, anchor=south] at (midpoint) {\pgfplotslegendfromname{grouplegendtypes-domino}};
\end{tikzpicture}

\caption{The effect of the instance type (Security, Availability or Noise) on the trade-off between $k_{\min}$ and $r$.
The time budget is 10\,000~sec.
Note that we added value $k_{\min} = 0$ to the vertical axis.}
\label{fig:real-instance-types}
\end{figure}
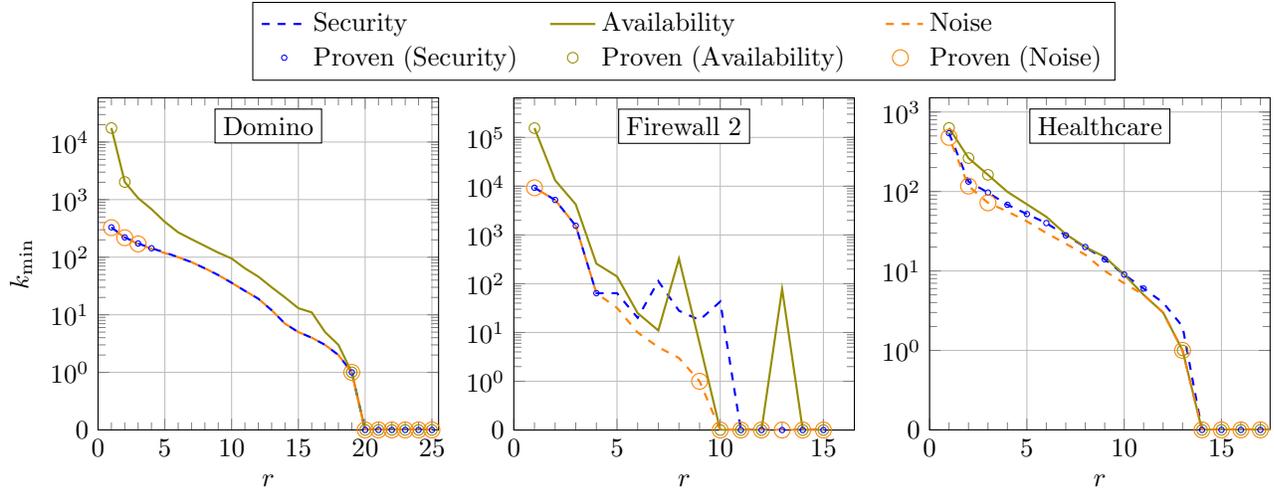

Noise instances are most flexible, hence, as expected, their $k$ does not exceed that of the Security and Availability instances.
The gap between Noise instances and Security/Availability instances depends on the $\bff{UPA}$ matrix.
For \emph{Domino}, the gap between Security and Noise instances is minimal (the ratio between their values of $k_{\min}$ does not exceed 1.025) whereas the gap between Availability and Noise instances is often significant (the ratio goes up to 53 for small $r$).
Intuitively, this makes sense considering that the \emph{Domino} $\bff{UPA}$ matrix is very sparse; to ensure that every value `1' is preserved, one may need to sacrifice many `0's.
For the \emph{Healthcare} instances, we observe that solutions to both Security and Availability instances are relatively close to the solutions to the Noise instances; the ratio never exceeds 2.3.
This is consistent with the observation that the \emph{Healthcare} $\bff{UPA}$ matrix is relatively dense.
It is difficult to make any conclusions for the \emph{Firewall~2} $\bff{UPA}$ matrix as some of the solutions are likely to be relatively far from optimal; this is evident from the jumps in the Security and Availability curves.


Finally, judging by the ability of the solver to prove optimality within the 10\,000 seconds time budget, the instances with small $r$ or $k_{\min}$ are easier than the instances with both $r$ and $k_{\min}$ being relatively large.
In other words, we found that the solver performs well on instances with either of the parameters being small whereas our theoretical results suggest that parameterisation by either of them is not efficient.
Our interpretation of these results is that the problem is hard even for small $k_{\min}$ or $r$ but it is particularly hard (from the practical point of view) when both $k_{\min}$ and $r$ are large.

\section{Concluding Remarks}
\label{sec:conclusion}

This paper introduces the {\sc Generalized Noise Role Mining} problem (GNRM).
We believe this is a useful contribution to the literature on role mining, not least because it allows us to define security- and availability-aware role mining problems.
We extended GNRM to BO-GNRM, a bi-objective optimization variant of GNRM that allows us to find an optimal balance between the number of roles $r$ and the number $k$ of  deviations from the
original authorization matrix.

We have shown that GNRM and BO-GNRM are fixed-parameter tractable, which means that they can be solved in a reasonable amount of time, provided problem instances only require solutions in which $r+k$ ($\bar{r}+\bar{k},$ respectively) are relatively small.
Algorithms for role mining do not need to be particularly fast, but they cannot be exponential in the size of the input, given the size of typical instances.
Knowing that algorithms exist that do solve role mining problems, subject to certain constraints on the solution, provides grounds for cautious optimism about the feasibility of solving real-world role mining problems.

Our experimental work provides further cause for optimism.
In particular, we show that our solution method based on a general-purpose solver efficiently exploits the FPT structure of the problem allowing us to solve several real-world instances, sometimes even proving solution optimality.
Also, the results on security-aware role mining suggest that it is possible to find solutions relatively quickly if we require those solutions to preserve the security of the original configuration, something that is generally desirable.

The work in this paper provides plenty of scope for future work. 
In particular, we would like to explore whether our work on matrix decomposition and FPT results can be extended to other variants of role mining (see~\cite[Section 2]{MitraSVA16}).
More generally, this paper along with~\cite{CramptonGKW17,CurreyMDG20}
demonstrates usefulness of multi-objective optimization in access control. Thus, it may be useful to  apply multi-objective optimization to other problems in the area.

\bibliographystyle{acm}

\bibliography{references}

\end{document}